\documentclass[11pt]{article}
\usepackage[utf8]{inputenc}

\usepackage{amsmath,amsthm,amssymb,xcolor,fullpage,hyperref,bbm}
\usepackage[capitalize,noabbrev]{cleveref}
\usepackage{algorithm, algorithmicx, algpseudocode}
\usepackage{enumerate}

\newtheorem{lemma}{Lemma}

\newtheorem{theorem}{Theorem}

\newtheorem{claim}{Claim}

\DeclareMathOperator*{\E}{\mathbb{E}}
\DeclareMathOperator*{\Var}{\mathrm{Var}}

\newcommand{\eps}{\epsilon}

\newcommand{\N}{\mathbb{N}}
\newcommand{\x}{\mathbf{x}}
\newcommand{\p}{\mathbf{p}}
\newcommand{\1}{\mathbbm{1}}
\newcommand{\poi}{\text{Poi}}
\newcommand{\poly}{\text{poly}}
\newcommand{\polylog}{\text{polylog}}
\newcommand{\DP}{\mathtt{DP}}
\newcommand{\profilealg}{\mathtt{EstimateProfile}}
\newcommand{\postprocess}{\mathtt{InvertCounts}}
\newcommand{\samplealg}{\mathtt{Sample}}
\newcommand{\updatealg}{\mathtt{IncrementCounters}}

\newcommand{\Justin}[1]{\textcolor{red}{Justin: #1}}
\newcommand{\Piotr}[1]{\textcolor{blue}{Piotr: #1}}

\title{Space-Optimal Profile Estimation in Data Streams with Applications to Symmetric Functions}
\author{
Justin Y.\ Chen\\MIT\\\texttt{justc@mit.edu}
\and Piotr Indyk\\MIT\\\texttt{indyk@mit.edu}
\and David P.\ Woodruff\\CMU\\\texttt{dwoodruf@andrew.cmu.edu}
}
\date{}

\begin{document}
\begin{titlepage}
\maketitle
\thispagestyle{empty}
\begin{abstract}
We revisit the problem of estimating the profile (also known as the rarity) in the data stream model. Given a sequence of $m$ elements from a universe of size $n$, its profile is a vector $\phi$ whose $i$-th entry $\phi_i$ represents the number of distinct elements that appear in the stream exactly $i$ times.  A classic paper by Datar and Muthukrishan from 2002 gave an algorithm which estimates any entry $\phi_i$ up to an additive error of $\pm \epsilon D$ using $O(1/\epsilon^2 (\log n + \log m))$ bits of space, where $D$ is the number of distinct elements in the stream. 

In this paper, we considerably improve on this result by designing an algorithm which simultaneously estimates many coordinates of the profile vector $\phi$ up to small overall error.
We give an algorithm which, with constant probability, produces an estimated profile $\hat\phi$ with the following guarantees in terms of space and estimation error:
\begin{alphaenumerate}
    \item For any constant $\tau$, with $O(1 / \epsilon^2 + \log n)$ bits of space, $\sum_{i=1}^\tau |\phi_i - \hat\phi_i| \leq \epsilon D$.
    \item With $O(1/ \epsilon^2\log (1/\epsilon)  + \log n + \log \log m)$ bits of space,  $\sum_{i=1}^m |\phi_i - \hat\phi_i| \leq \epsilon m$.
\end{alphaenumerate}
In addition to bounding the error across multiple coordinates, our space bounds separate the terms that depend on $1/\epsilon$ and those that depend on $n$ and $m$. We prove matching lower bounds on space in both regimes.

Application of our profile estimation algorithm gives estimates within error $\pm \epsilon D$ of several symmetric functions of frequencies in $O(1/\epsilon^2 + \log n)$ bits. This generalizes space-optimal algorithms for the distinct elements problems to other problems including estimating the Huber and Tukey losses as well as frequency cap statistics.

\end{abstract}
\end{titlepage}

\section{Introduction}
\label{s:intro}
Estimating basic statistics of a data set, such as the number of times each element occurs or the number of distinct elements are fundamental problems in data stream algorithms.
In this paper, we focus on the related problem of estimating the number of elements that occur a given number of times.
Formally, given a stream $\x = x_1, \ldots, x_m$ of $m$ elements from $[n] = \{1, 2, \ldots, n\}$, we define its {\em profile} to be the frequency of frequencies vector $\phi \in \{0, \ldots, n\}^m$ where $\phi_i = |\{j \in [n]: |\{k \in [m]: x_k = j\}| = i\}|$ is the number of distinct elements in $\x$ which appear exactly $i$ times; such elements are often referred to as ``$i$-rare'' and were first studied in a classic paper of Datar and Muthukrishnan~\cite{datar2002estimating}.

Any \emph{symmetric} function of the frequencies of the elements in a stream (such a function is invariant to relabeling of the domain elements) can be written as a function of the profile.
Therefore, algorithms for profile estimation in data streams can be used to estimate quantities such as the number of distinct elements~\cite{flajolet1985probabilistic}, frequency moments~\cite{alon1996moments}, capped statistics of the stream~\cite{cohen2015stream}, or the objective function of $M$-estimators such as the Huber or Tukey  objective~\cite{jayaram2022samplers}.
In this paper, we develop an algorithm for estimating the profile of a data set specified by an insertion-only stream, where elements are inserted but not deleted.
As an application, this algorithm improves upon the space complexity of estimating several symmetric functions of frequencies.

The profile (also referred to as the fingerprint, histogram, histogram of histograms, pattern, prevalence, or collision statistics) of a data set is a natural representation of the distribution of elements and has been studied extensively from both computational and statistical perspectives. Streaming algorithms that estimate the number of $i$-rare elements have been used for computing degree distributions in large graphs~\cite{buriol2005using}, detecting malicious IP traffic in a network~\cite{karamcheti2005detecting}, estimating the number of times users have been exposed to the same ad~\cite{ghazi2022multipartyreach}, counting the number of $k$-mers in genetic sequences with a given abundance value for fast $k$-mer size selection \cite{chikhi2014informed}, and for applications in databases~\cite{cormode2005summarizing}.
In practice, estimating the profile is a very popular sketching problem solved by users of Apache DataSketches, a popular open-source sketching library~\cite{dickens2023datasketchesemail}\footnote{See this example of estimating the number of estimating the distribution of how many times users visit a website in a month \url{https://datasketches.apache.org/docs/Tuple/TupleEngagementExample.html}.}.


The study of the problem in the context of streaming algorithms dates back to the work of Datar and Muthukrishan \cite{datar2002estimating}. 
They show how to estimate the ratio of $i$-rare elements in the stream  to the total number of distinct elements $D$ in the stream, i.e., the fraction $\phi_i/D$. The algorithm is simple and elegant: it collects a random sample $j_1 \ldots j_s$ of $s$ elements from the stream, where each $j_t$ is chosen uniformly at random from the set of distinct elements appearing in the stream. For each element $j_t$, it calculates the element's frequency (the number of times $j_t$ appears in the stream) and returns the fraction of $j_t$'s with frequency exactly $i$. The algorithm can be implemented in one pass using $O(s (\log n + \log m))$ bits of storage, and the authors show a trade-off between the quality of approximation and the sample size $s$. Specifically, to approximate $\phi_i/D$ up to $\pm \eps$ with constant probability, it suffices that $s=\Theta(1/\eps^2)$\footnote{The original paper provides a more refined bound with a mix of additive and multiplicative errors, see Lemma 2 in \cite{datar2002estimating}. We provide a single-parameter bound for the sake of simplicity.}. This translates into an $O(1/\eps^2 (\log n + \log m))$ space streaming algorithm which, given a particular $i$, finds an estimate $\hat{\phi}_i$ such that
\begin{equation}
\label{e:l_infty}
|\phi_i-\hat{\phi}_i|  \le \eps D. 
\end{equation}

Other works have studied rarity estimation streaming algorithms in the context of the sliding window model~\cite{braverman2014catch}, time and space efficient algorithms~\cite{feigenblat2011exponential} and privacy~\cite{dwork2010pan}. 
The ``layering'' technique of Indyk and Woodruff~\cite{indyk2005optimal}  essentially estimates the number of items whose frequencies are {\em approximately} equal to a given value.
A pertinent line of work studies algorithms for estimating general ``concave sublinear'' frequency statistics which depend on the rarities of low frequency elements~\cite{cohen2017hyperloglog, cohen2019samplingsketches}.
These papers provide succinct sketches for estimating several symmetric functions of frequencies but either do not account for the space used to store hash functions or do not focus on bit complexity and require two passes over the data.
When the space used for randomness is attributed to the algorithm, our results give improved space bounds compared to prior work for estimating several statistics that fall into this framework (see \cref{subsec:applications}).


\subsection{Results for Profile Estimation}
We focus on estimating several coordinates of the profile vector $\phi$ simultaneously with small {\em total} error. In many applications, it is useful to calculate more than one rarity as the true object of interest is the distribution of frequencies.
In addition, estimating several coordinates of the profile has direct applications to estimating several symmetric functions of frequencies (see \cref{subsec:applications}).
Our main result is a streaming algorithm which, with constant probability, achieves the following guarantees (Theorems~\ref{thm:algD} and~\ref{thm:algm}):
\begin{itemize}
    \item For any $\tau = O(1)$, using space $O(1 / \eps^2 + \log n)$, the algorithm returns a $\hat\phi$ such that
    \begin{equation}\label{e:l1_D}
        \sum_{i=1}^\tau |\phi_i - \hat\phi_i| \leq \eps D.
    \end{equation}

    \item Using space $O(1 / \eps^2 \log (1/\eps)  + \log n + \log \log m)$, the algorithm returns a $\hat\phi$ such that
    \begin{equation}\label{e:l1_m}
        \sum_{i=1}^m |\phi_i - \hat\phi_i| \leq \eps m.
    \end{equation}
\end{itemize}

Both results use less space than the algorithm of Datar and Muthukrishnan while bounding error across several coordinates of the profile rather than for a single rarity.
A brief remark comparing the two guarantees: as $D \leq m$, the error in \cref{e:l1_D} is smaller than that of \cref{e:l1_m} but at the cost of only providing a guarantee for estimating the profile over constant frequencies.\footnote{Note that in this regime, to estimate the entire profile up to the $\tau$th coordinate, it is sufficient to have an algorithm for estimating a single $i$-rarity for constant $i$ and making a constant number of copies of this algorithm for each $i \in \{1, \ldots, \tau\}$.
Interestingly, our algorithm internally produces an estimate of the entire profile up to $i$ even if the goal is only to estimate the single $i$-rarity.}
On the other hand, by \cref{e:l1_m} we can estimate the entire profile in small space up to error $\pm \eps m$. Note that in common settings where the average frequency of the elements in the stream is small, $\eps m$ and $\eps D$ are of similar size. However, in general, estimating the entire profile up to additive error in terms of $D$ is hard. In fact, producing an estimate such that $\|\phi - \hat \phi\|_1 \leq D/2$ would require $\Omega(D\log m)$ bits of space.\footnote{Let $S_N^D$ be the set of all  $D$-sparse binary vectors of length $N$, and let $C \subset S_N^D$ be its subset such that any pair of distinct $c,c' \in C$ have $L_1$ distance greater than  $D$. Standard probabilistic arguments show that there exists such a set $C$ of size $\exp(\Omega(D \log (N/D)))$. Observe that for any $c \in C$, we can generate a stream with profile equal to $c$, by creating,  for each nonzero $c_i$,  a distinct element appearing $i$ times. 
Given such a stream, the assumed algorithm returns $\hat{\phi}$ with $L_1$ distance at most $D/2$ to $c$, which makes it possible to uniquely recover $c$. By the pigeonhole principle, the algorithm must use space at least $\Omega(\log |C|) = \Omega(D \log (N/D))$. Since the length of the generated stream $m$ is at most $ND$, if we pick $N>D^2$, we obtain the desired bound.} 

The first step in our analysis is to show that $s = O(1/\eps^2 \log (1/\eps) )$  samples suffice to achieve the guarantees of \cref{e:l1_m} using the empirical estimation method of Datar and Muthukrishnan ($s = O(1/\eps^2)$ suffices for \cref{e:l1_D} from the original paper). However, this na\"ively requires $s(\log n + \log m)$ space to store the identities and counts of each sample.

Our main result reduces the space required to $O(s)$ bits. The algorithm is obtained by compressing the identity and count information of each sample into $O(1)$ bits on average while retaining its statistical power for estimating the profile.
The algorithm is also time efficient: it takes $O(\log(1/\eps) + \log\log n)$ expected amortized time to process each stream update and $\poly(1/\eps)$ time to produce the final profile estimate.

To achieve these results, we use {\em two new techniques}:
\begin{enumerate}
\item To compress the identities of the stream elements, we hash the sampled elements to a domain of size $O(s)$, allowing collisions between sampled elements in a similar manner to Bloom filters or CountMin sketches. Our key contribution is to show that, under the parameters of our algorithm, the empirical profile of the sampled elements can be approximately recovered from the frequencies after hashing by an iterative ``inversion'' procedure. To our knowledge, this inversion procedure is novel and requires a careful analysis of the hashing procedure as well as an application of the ``Poissonization'' trick more commonly used in distribution testing.

\item To use small space, we need to be able to make efficient use of randomness for our hash functions. A key statistic used in our analysis is the number of buckets in the hash table with frequency $i$ for $i \in [m]$. Our analysis requires $O(1)$-wise independence of the associated random variables, however, this does not simply follow from using an $O(1)$-wise independent hash family. On the other hand, Nisan's generator can be used in a black-box fashion~\cite{indyk2006stable}, but this would blow up the space bound by a logarithmic factor. Instead, we apply Nisan's generator to a subroutine of the streaming algorithm to ensure that $O(1)$-wise independence holds for the pertinent random variables. To our knowledge, this is a novel technique and one that seems quite versatile: since its introduction in an earlier version of this paper, it has been already used for other streaming problems~\cite{kacham2023pseudorandom}.
\end{enumerate}

We complement our algorithmic results with the following lower bounds which show that we achieve the \emph{optimal} dependence on the error parameter $\eps$ (Theorems~\ref{thm:lb_m} and~\ref{thm:lb_D}).
\begin{itemize}
    \item Any one-pass algorithm satisfying \cref{e:l1_D} with constant probability must use at least $\Omega(1/\eps^2)$ bits of space.
    
    \item Any one-pass algorithm satisfying \cref{e:l1_m} with constant probability must use at least $\Omega(1/\eps^2 \log (1/\eps) )$ bits of space. 
\end{itemize}
To the best of our knowledge, the latter is a rare example of a natural streaming problem where the optimal dependence of the space bound on the accuracy parameter $\eps$ is not of the form $1/\eps^a$ for some integer exponent $a \ge 1$.

\subsection{Applications to Symmetric Functions}
\label{subsec:applications}

By itself, the profile is a useful statistic of the stream, but it is also important in that any symmetric (invariant to relabeling) function of frequencies can be written as a function of the profile.
Therefore, the guarantees of the algorithm given in \cref{e:l1_D} and \cref{e:l1_m} can be leveraged in a black-box way to give streaming algorithms for estimating a variety of symmetric functions of the frequencies of the stream.
We give several illustrative examples where we can estimate functions in essentially the same space required to estimate the number of distinct elements. In what follows, consider constant $\tau$.
\begin{itemize}
    \item \emph{Distinct elements with frequency at most or at least $\tau$:} The number of distinct elements with frequency at most $\tau$ is the sum of the first $\tau$ coordinates of the profile and can be calculated up to $\pm \eps D$ in $O(1/\eps^2 + \log n)$ space using our algorithm. The number of distinct elements with frequency at least $\tau$ can be calculated by subtracting those with frequency at most $\tau-1$ from the total number of distinct elements which can be also be approximated in $O(1/\eps^2 + \log n)$ space~\cite{kane2010optimal}.
    \item \emph{Mass of elements with frequency at most or at least $\tau$:}
    The mass of the distinct elements with frequency at most $\tau$ can be expressed as $\sum_{i=1}^\tau \phi_i \cdot i$ which can be calculated in space $\pm \eps D$ in $O(1/\eps^2 + \log n)$ space using our algorithm. Subtracting from the total mass of the stream which can be approximated up to $\pm \eps m$ in space $O(\log \log m + \log(1/\eps))$ with a Morris counter~\cite{morris1978counting} yields the mass of elements with frequency at least $\tau$.
    \item \emph{Capped (or saturated) statistics~\cite{cohen2015stream}:} For a given parameter $\tau$, the corresponding capped statistic of the stream is 
    \begin{equation*}
        \sum_{i=1}^\tau \phi_i \cdot i + \sum_{i=\tau+1}^m \phi_i,
    \end{equation*}
    a generalization of counting the stream length and counting the number of distinct elements. This can be calculated using the two quantities above up to error $\pm \eps D$ in $O(1/\eps^2 + \log n)$ space. 
    \item \emph{Tukey objective~\cite{Rey1983robust}} For a given parameter $\tau$, the Tukey objective is 
    \begin{equation*}
        \sum_{i=1}^\tau \phi_i  \cdot \frac{\tau^2}{6}\left(1 - \left(1 - i^2/\tau^2\right)^3\right) + \sum_{i=\tau + 1}^n \phi_i \cdot \frac{\tau^2}{6}.
    \end{equation*}
    The first summation can be estimated up to $\pm \eps D$ using our algorithm in $O(1/\eps^2 + \log n)$ space. The second summation is $\tau^2/6$ times using the number of distinct elements with frequency at least $\tau +1$ and thus can also be estimated up to $\pm \eps D$ in $O(1/\eps^2 + \log n)$ space.
    \item \emph{Huber objective~\cite{huber1964robust}} For a given parameter $\tau$, the Huber objective is 
    \begin{equation*}
        \sum_{i=1}^\tau \phi_i  \cdot \frac{i^2}{2} + \sum_{i=\tau + 1}^n \phi_i \cdot (\tau i - 1/2).
    \end{equation*}
    The first summation is can be estimated up to $\pm \eps D$ using our algorithm in $O(1/\eps^2 + \log n)$ space. The second summation can be written as $\tau$ times the mass of elements with frequency at least $\tau+1$ minus half the number of distinct elements with frequency at least $\tau + 1$. So, in total, the objective can be estimated up to $\pm \eps D$ in $O(1/\eps^2 + \log n + \log\log m + \log(m/D))$ space.
\end{itemize}

To our knowledge, the best known previous algorithms for these problems use $O(1/\eps^2 \log n)$ bits of space (by storing identities of sampled elements) or do not account for the space associated with randomness~\cite{datar2002estimating, cohen2015stream, cohen2017hyperloglog, cohen2019samplingsketches}.
In general, we can apply the guarantees of \cref{e:l1_D} to estimate, in $O(1/\eps^2 + \log n)$ space, any symmetric function which depends on a constant frequencies and is Lipschitz with respect to $L_1$ error in the profile.

\subsection{Technical Overview}
For simplicity, we focus in this overview on the $\pm \eps m$ guarantee of the algorithm in \cref{e:l1_m}. The same algorithm, but with slightly different parameters, achieves the guarantee in \cref{e:l1_D}.
Our new analysis of the algorithm of Datar and Muthukrishnan is given in \cref{s:basic}. Although the algorithm is suboptimal, it illustrates the core issues that the techniques of optimal algorithm have to address.  The method samples elements uniformly from the set of distinct elements in the stream and uses the (rescaled) empirical profile of the samples as the estimated profile $\hat{\phi}$.
In the context of an $L_1$ guarantee of $\eps m$ additive error, we observe that it suffices to estimate only the $\phi_i$'s for $i$ up to $O(1/\eps)$, as the remaining values can be set to zero without incurring much error (as there cannot be many high frequency elements). We then study the expected $L_1$ estimation error of the first $O(1/\eps)$ entries. The analysis crucially uses the specific properties of the profile function, leading to an $O(1/\eps^2 \log(1/\eps))$ bound on the sample size. Na\"ively, each of these sampled elements requires $\log n + \log m$ bits to store its identity and count.

This algorithm can be improved using the fact that, to compute the profile, the actual identities of the sampled elements are not important as long as we can distinguish among them. This makes it possible to reduce the space by hashing sampled elements to a smaller universe of size that is quadratic in the sample size (quadratic dependence being necessary to avoid collisions). Since the sample size is polynomial in $1/\eps$, each hash can be represented using $O(\log (1/\eps))$ bits. As we are only concerned with frequencies up to $O(1/\eps)$, the counts of each sampled element can also be stored in $O(\log(1/\eps))$ bits, leading to an overall space bound of $O(1/\eps^2 \log^2(1/\eps))$ bits. Although this algorithm is suboptimal, we believe that its simplicity makes it appealing in practice. 

The optimal algorithm (\cref{s:algorithm}) is much more technically involved. As with the previous algorithm, it hashes sampled elements into a smaller universe to reduce space. However, the size of the hash table is now linear, not quadratic, in the sample size. This removes the need to store the hashed identities as we can store the entire hash table explicitly. Combined with a more careful analysis of the number of bits required to represent the counts of sampled elements, this removes the ``extra’’ $\log (1/\eps)$ factor. This improvement, however, comes at the price of allowing collisions, meaning that elements with different frequencies are now mixed together, and the profile of hashed elements does not approximate the original one\footnote{For example, a stream of all distinct elements will be hashed to one where a constant fraction of elements will have duplicates.}. This necessitates inverting this mixing process to obtain frequency estimates for the original sample.

To simplify the analysis, we use the ``Poissonization'' trick so that outcomes in different buckets in the hash table are independent. Specifically, we use an additional hash function which maps an element to a $\poi(1)$ random variable. We create that many distinct copies of each sampled element and add these copies to the hash table.
Our goal is to use the hash table to estimate the number of sampled elements (before Poissonization) with frequency $i$ for $i \in \{1, \ldots, O(1/\eps)\}$.
We achieve this via an iterative algorithm. Letting $\hat{F}_j$ be our estimate for the number of distinct elements in our sample with frequency $j$, assume we are given the estimates $\hat{F}_1, \ldots, \hat{F}_{i-1}$ and want to estimate $\hat{F}_i$.
We observe that there are two types of buckets in the hash table with count $i$.
``Good'' buckets are those which contain a single element with frequency $i$.
``Bad'' buckets are those which contain multiple elements (due to hash collisions) which sum to $i$.

To estimate the number of bad buckets, we sum, over all integer partitions of $i$ with at least two summands, the estimated probability that that exact combination of elements hashes to the same bucket. These estimated probabilities come from our estimates $\hat{F}_1, \ldots, \hat{F}_{i-1}$ as well as our knowledge of the sample size compared to the size of the hash table.
We can then estimate the number of good buckets by subtracting the estimated number of bad buckets from the total number of observed buckets with count $i$.
Finally, we estimate $\hat{F}_i$ by inverting the probability of getting a good bucket, i.e., of a bucket containing exactly one element.

This procedure produces the correct estimates under the assumption that the numbers of each type of bucket described above occur according to their expectations.
As this is not the case, two types of error are introduced when calculating $\hat{F}_i$: random error due to deviations of the number of buckets with count $i$ and propagation error due to using noisy estimates $\hat{F}_1, \ldots, \hat{F}_{i-1}$ in the calculation of the number of bad buckets with count $i$.
At first glance, this second type of error has the potential to grow out of control as errors early on compound through the iterative estimation procedure.
However, we show that the total propagation error across all coordinates of the estimated profile is within a constant factor of the total random error.
To prove this result, we carefully analyze the sensitivity of the estimation of the number of bad buckets with count $i$ to changes in the estimated numbers of elements with counts less than $i$.
Early errors can compound as they recursively affect all subsequent estimates: error in estimating the number of elements with frequency $i$ affects the estimate of the number of bad buckets for all counts greater than $i$. However, the propagation of these errors is limited by the fact that the probability of a large number of elements hashing to the same bucket decays exponentially.
We ultimately bound the error introduced by allowing hash collisions to an additional $O(\eps m)$ term in the expectation of $\|\phi - \hat{\phi}\|_1$.

One aspect of the algorithm that we have so far swept under the rug is how the algorithm samples elements: we need to adaptively maintain a sample of $O(1/\eps^2 \log(1/\eps))$ elements. In a similar methodology to the optimal distinct elements sketch~\cite{kane2010optimal}, we hash each element to a random identity in $[n]$ and sample all elements which have least significant bit at least $\ell$ after hashing. The variable $\ell$ indicates the current sampling ``level''. We track the stream length and number of distinct elements over time in order to update the level and maintain the correct number of samples. In order to remove the counts associated with elements that are no longer in the sample once the level updates, in each cell of the hash table, we keep separate counters stratified by the least significant bit of the hash of contributing elements.

The analysis of this algorithm requires pairwise independence of the counts of buckets. To ensure this pairwise independence holds after replacing truly random bits by a pseudorandom generator, we use Nisan's pseudorandom generator. Since we need to preserve distributions of the counts of {\em pairs} of buckets, which are $O(\log(1/\eps))$-bit long, a random seed of length $\polylog(1/\eps)$ suffices.  (Note that we cannot use Nisan's generator to ensure that the bucket counts are fully independent, as that would require a random seed of length equal to the number of buckets, the space of our algorithm, times $\log(1/\eps)$). We note that the technique of employing Nisan's generator to achieve $O(1)$-wise independence introduced in this paper appears to be quite versatile, and has been since used for other streaming algorithms~\cite{kacham2023pseudorandom}.

Our lower bound (\cref{s:lower}) for algorithms achieving the guarantee of \cref{e:l1_m} proceeds via a reduction from a direct sum of multiple instances, where each instance can be viewed as the composition of the Indexing problem with the Gap Hamming Distance problem with different parameters.
To illustrate the basic connection between profile estimation and these communication problems, note that $\phi_1$ can be used to count the number of elements which appear in exactly one of two binary strings to solve the Gap Hamming Distance problem\footnote{In fact, this simple reduction is how we prove a lower bound for algorithms achieving the guarantee of \cref{e:l1_D}.} while distinguishing between there existing an element with frequency $i$  or $i-1$ can be used to solve Indexing (by Bob adding $i-1$ copies of the element corresponding to his index).

The entries $\phi_1, \phi_2 \ldots \phi_{1/\eps}$ of the profile vector are split into ``scales'', where each scale contains entries $\phi_i$ for comparable (up to a constant factor) values of $i$. Intuitively, each scale contributes $1/\eps^2$ term to the lower bound, for a total bound of $\Omega(1/\eps^2 \log(1/\eps))$ bits. As there are known reductions of Gap Hamming Distance from the Indexing problem, we ultimately are able to prove our entire lower bound via an involved reduction from Indexing itself.

\subsection{Discussion and Open Questions}
We revisit the problem of estimating the profile of a data stream, a problem that appears commonly in practice and has applications to estimating symmetric functions of frequencies. 
We give space-optimal algorithms for two types of error guarantees. 
Our results focus on producing good estimates for entries of the profile corresponding to elements with small frequency (either explicitly through the parameter $\tau$  in \cref{thm:algD} or implicitly by letting the error scale with the mass of the stream in \cref{thm:algm}).

One direction for future work is to study profile estimation guarantees that put more emphasis on estimating large entries of the profile.
What is the optimal space complexity of estimating the profile up to $\pm \eps D$ on the first $\tau$ coordinates for superconstant $\tau$? Recall that if $\tau > D^2$, $\Omega(D \log m)$ bits of space are required.
Estimation in terms of $L_1$ error of the profile requires that we approximate the number of elements appearing exactly $i$ times (for many $i$).
If we allow for approximating the number of elements appearing \emph{approximately} $i$ times, can we use less space?
Answering these questions may imply improved algorithms for a broader class of symmetric functions.

The profile also appears in literature on distribution testing. Several works use the profile of a sample from a distribution to give sample-optimal testers for a broad class of symmetric properties (e.g., testing uniformity or estimating entropy)~\cite{valiant2017estimating, acharya2016unified, charikar2019efficient, anari2021bethe}. For the right notion of error, a streaming algorithm for profile estimation may be able to be used to process a sample in sublinear space while retaining the performance of the testing algorithms. We leave the study of this as an intriguing open question. 

\paragraph*{Paper Organization} 
The paper is organized as follows. In \cref{s:algorithm}, we present and analyze our space-optimal algorithm. In \cref{s:lower}, we present the lower bounds, showing the optimality of our algorithm. Finally, in \cref{s:basic}, we include the analysis of a simpler but suboptimal algorithm based on that of Datar and Muthukrishnan.

\section{Profile Estimation Algorithm}
\label{s:algorithm}
\begin{theorem}[$\pm \eps D$] \label{thm:algD}
For any $\eps > 0$ and $\tau = O(1)$, with input parameters $B = \Theta(1 / \eps^2)$ and $errortype = D$, \cref{alg:profile} uses $O(1 / \eps^2 + \log n)$ bits of space, $O(\log(1/\eps) + \log\log n)$ expected amortized update time, $O(1)$ post-processing time, and returns an estimated profile $\hat{\phi}$ that satisfies 
\begin{equation*}
    \sum_{i=1}^\tau |\phi_i - \hat{\phi}_i| \leq \eps D
\end{equation*} with probability $9/10$.
\end{theorem}

\begin{theorem}[$\pm \eps m$] \label{thm:algm}
For any $\eps > 0$, with input parameters $B = \Theta(1/ \eps^2 \log (1/\eps) )$, $\tau = O(1/\eps)$, and $errortype = m$, \cref{alg:profile} uses $O(1 / \eps^2 \log (1/\eps)  + \log n + \log\log m)$ bits of space, $O(\log(1/\eps))$ expected amortized update time, $O(1/\eps^3 \log (1/\eps) )$ post-processing time, and returns an estimated profile $\hat{\phi}$ that satisfies 
\begin{equation*}
    \sum_{i=1}^m |\phi_i - \hat{\phi}_i| \leq \eps m
\end{equation*} with probability $9/10$.
\end{theorem}

Before we describe and analyze \cref{alg:profile}, we will give a few remarks on its inputs. In addition to the stream, the algorithm takes as input several parameters: the domain size $n$ of the stream elements, a frequency threshold $\tau$ (we will ignore counts that exceed this threshold), a number of buckets $B$ for the core hash table, an error parameter $\eps$, and a variable $errortype$ indicating whether the error guarantee will be $\pm \eps D$ or $\pm \eps m$.

For the input parameter $n$, only an upper bound on the domain size is required to set the domain and range of the hash function $g_1$ which is used to sample elements. As our bounds depend logarithmically in $n$, any $\poly(n)$ upper bound suffices and we will assume for simplicity it is a power of two.
We will assume that $B = \Omega(\log n)$. If this is not the case, we can pick a smaller $\eps$ so that $B = \Theta(\log n)$, paying an additive term of $O(\log n)$ in space and leaving the asymptotic complexity unchanged. We also assume that $B = O(D)$: otherwise we have almost enough space to store the entire frequency histogram of the stream.

\begin{algorithm}[ht!]
\caption{\label{alg:profile} $\profilealg$}
\begin{flushleft}
{\bfseries Input:} stream $\mathbf{x} = x_1, \ldots, x_m$, domain size $n$, frequency threshold $\tau$, number of buckets $B$, error parameter $\eps$, $errortype$ (either $D$ or $m$) \\
{\bfseries Output:} estimated profile $\hat{\phi} = \hat{\phi}_1, \ldots \hat{\phi}_\tau$  \\
\begin{algorithmic}[1]
\State $T \gets \Theta(B^2), H \gets \Theta(\log B), K \gets \Theta(1)$
\State Initialize a distinct elements sketch with relative error $\eps/10$~\cite{kane2010optimal} and, if $errortype=D$, also initialize a strong tracking distinct elements sketch with relative error $1/10$ ~\cite{blasiok2020distinctelements}
\State Initialize the main array $A$ of $B$ buckets as a variable-bit-length array~\cite{blandford2008CompactDF} which will store in each bucket a string of (level, counter) pairs
\State Initialize hash functions $g_1: [n] \rightarrow [n]$, $g_2: [n] \rightarrow [T]$, $h_1, \ldots, h_H: [T] \rightarrow [B]$
\State Initialize hash function $z: [T] \rightarrow \N \cup \{0\}$ that maps values in $[T]$ to the outcome of a $\poi(1)$ random variable
\State Initialize current level $\ell_{cur} \gets 1$
\For{$x_t \in \mathbf{x}$} \Comment{Processing stream updates}
    \State Update distinct elements sketches
    \State Let $\tilde{D}_t$ be the estimate of the tracking sketch and let $t$ be the stream length so far
    \If{($errortype = m$ and $2^{\ell_{cur}} < \min\{t K/B, n\}$) or \\
    \qquad ($errortype=D$ and $2^{\ell_{cur}} < \tilde{D}_t K/B$)} \Comment{Decrease sampling probability} \label{line:update-level}
        \State $\ell_{cur} \gets \ell_{cur} + 1$
        \State For each (level, counter) pair in $A$, decrement the level
        \State If any level falls below $0$, remove the corresponding pair from $A$
    \EndIf
    \State $\ell, a_1, \ldots, a_H \gets \samplealg(x_t, H, g_1, g_2, z,h_1, \ldots, h_H)$
    \If{$\ell \geq \ell_{cur}$}
        \State $\updatealg(A, a_1, \ldots, a_H, \ell-\ell_{cur}, \tau)$
    \EndIf
\EndFor
\State $\hat{D} \gets$ distinct elements estimate with error $\eps/10$ \Comment{Post-processing to estimate profile}
\State $G \gets$ the number of nonempty buckets in $A$
\State $\hat{S} \gets - B \ln\left(1 - \frac{G}{B}\right)$ estimate of the number of elements which hash to $1$ under $g_1$
\State $b_i \gets $ number of buckets in $A$ with total count $i$ for $i \in \{1, \ldots, \tau\}$
\State $\hat{F}_1, \ldots, \hat{F}_{\tau} \gets \postprocess(B, \hat{S}, b_1, \ldots, b_{\tau})$ \Comment{Estimate the profile of the sampled elements}
\State $\hat{\phi}_i \gets \left(\frac{\hat{D}}{\hat{S}}\right) \hat{F_i}$ for $i \in \{1, \ldots, \tau\}$
\State \Return $\hat{\phi}_1, \ldots, \hat{\phi}_\tau$
\end{algorithmic}
\end{flushleft}
\end{algorithm}

\begin{algorithm}[ht!]
\caption{\label{alg:sampling} $\samplealg$}
\begin{flushleft}
{\bfseries Input:} stream element $x$, max copies $H$, hash functions $g_1, g_2, z, h_1, \ldots, h_H$ \\
{\bfseries Output:} sampling level, buckets to update $a_1, \ldots, a_H$
\begin{algorithmic}[1]
\State $\ell \gets$ least significant nonzero bit of $g_1(x)$
\State $x' \gets g_2(x)$ \Comment{Smaller ID}
\For{$i \in [H]$}
    \If{$i \leq z(x')$} \Comment{Checking number of copies from Poissonization}
        \State $a_i \gets h_i(x')$
    \Else
        \State $a_i \gets 0$
    \EndIf
\EndFor
\State \Return $\ell, a_1, \ldots, a_H$
\end{algorithmic}
\end{flushleft}
\end{algorithm}

\begin{algorithm}[ht!]
\caption{\label{alg:update} $\updatealg$}
\begin{flushleft}
{\bfseries Input:}  Array $A$, buckets $a_1, \ldots, a_H$ to increment, level $\ell$, max frequency $\tau$ \\
\begin{algorithmic}[1]
\State $j \gets$ smallest $i \in [H]$ s.t. $a_i = 0$ \Comment{The number of buckets to update}
\For{$i = 1, \ldots, j-1$}
    \If{there exists a (level, counter) pair with level $\ell$ in $A[a_i]$ }
        \State Increment the corresponding counter unless it exceeds $\tau$
    \Else
        \State Add a new pair ($\ell$, $1$) to $A[a_i]$
    \EndIf
\EndFor
\end{algorithmic}
\end{flushleft}
\end{algorithm}

\begin{algorithm}[ht!]
\caption{\label{alg:dp} $\postprocess$}
\begin{flushleft}
{\bfseries Input:} number of buckets $B$, estimated number of sampled elements $\hat{S}$, number of buckets $b_i$ with count $i$ for $i \in \{1, \ldots, \tau \}$  \\
{\bfseries Output:} estimated counts $\hat{F}_1, \ldots \hat{F}_{\tau}$  \\
\begin{algorithmic}[1]
\State Initialize $\tau \times \tau$ array $\DP$
\State $\DP[1,1] \gets b_1 e^{\hat{S}/B}$
\For{$i \in \{2, \ldots, \tau\}$}
    \State $\DP[i,i] \gets \max \{b_i e^{\hat{S}/B} - \sum_{x=1}^{\lfloor i/2 \rfloor} \DP[i,x], 0\}$
    \For{$x \in \{1, \ldots, \lfloor i/2 \rfloor\}$}
        \State $\DP[i,x] \gets \sum_{k=1}^{\lfloor i/x \rfloor - 1} \sum_{x' = x+1}^{i - kx} \DP[i - kx, x'] \left(\frac{\DP[x, x]}{B}\right)^k \frac{1}{k!}$
        \If{$i = 0 \bmod x$}
            \State $\DP[i,x] \gets \DP[i,x] + \frac{\DP[x, x]^{i/x}}{B^{i/x - 1} (i/x)!}$
        \EndIf
    \EndFor
\EndFor
\State $\hat{F}_i \gets \DP[i,i]$ for $i \in \{1, \ldots, \tau \}$
\end{algorithmic}
\end{flushleft}
\end{algorithm}

\paragraph*{Algorithm Description}
The algorithm is decomposed into four parts: the main algorithm $\profilealg$ (\cref{alg:profile}), the sampling procedure $\samplealg$ (\cref{alg:sampling}), the update procedure $\updatealg$ (\cref{alg:update}), and the post-processing procedure $\postprocess$ (\cref{alg:dp}).
The core data structure maintained by the algorithm is an array $A$ of $B$ buckets. Each bucket will contain, as necessary, pairs of (level, counter) indicating the summed frequency of items which hash to that bucket with a certain sampling ``level''.

Elements are sampled using the hash function $g_1$ and the $\samplealg$ subroutine. The main algorithm maintains a current level $\ell_{cur}$ and a stream element $x$ is sampled if the position of the least significant $1$ bit in the binary representation of $g_1(x)$ is at least $\ell_{cur}$. The level of the sampled element is the position of its least significant bit minus $\ell_{cur}$. $\poi(1)$ copies of sampled elements are made and assigned to random locations in the hash table using hash functions $g_2, z, h_1, \ldots, h_H$. The main algorithm periodically updates $\ell_{cur}$ to reduce the sampling probability based on constant factor estimates of the number of distinct elements and stream length to ensure that the number of samples is correct.

For each copy of a sampled element, we update the bucket count in the corresponding bucket in the hash table via $\updatealg$. If a (level, counter) pair already exists in that bucket for the level of our sample, we increment the counter; otherwise, we create a new (level, counter) pair. If a counter ever exceeds $\tau$, we stop incrementing that counter. When the current level $\ell_{cur}$ is incremented, we update all (level, counter) pairs in all buckets of the array by decrementing the level (remember the level of a sample is the position of the least significant bit relative to $\ell_{cur}$). Any time the level goes below zero, we remove the corresponding pair from its bucket.

At the end of the stream, we observe the number of buckets with total counts (summed over all counters in the bucket) $1, \ldots, \tau$, i.e., the profile of the array. In $\postprocess$, we estimate the profile of the sampled elements (which have been corrupted by hash collisions) in an iterative process using dynamic programming. Using the number of nonempty bins, we estimate the probability that a bucket receives a single element. We call these ``good'' buckets as their counts correspond to a single sampled item of that frequency. Using this estimated probability, we can estimate the number of items with a given frequency from the number of observed good buckets with that count. Unfortunately, there also exist ``bad'' buckets with multiple items whose count is the sum of the frequency of those items. As there cannot be a bad bucket with count $1$, we first estimate the number of sampled elements with frequency $1$ and use that to estimate the number of bad buckets of frequency $2$. These estimates are then used to estimate the number of bad buckets of frequency $3$, and so on. The dynamic program allows us to efficiently compute these iterative estimates without having to exhaustively list all integer partitions. Our final profile estimate in $\profilealg$ comes from renormalizing the estimates of the empirical sample profile returned by $\postprocess$.

\paragraph*{Notation}
To introduce some notation, let $D_t = |\{x \in (x_1, \ldots, x_t)\}|$ be the number of distinct elements in the stream up to time $t \in [m]$ with the shorthand $D = D_m$. Let $\tilde{D}_t$ be the estimate of the tracking sketch at time $t$. Let $\ell_t$ be the value of $\ell_{cur}$, the current level of the algorithm, at time $t$.

Let $S = \{x \in \mathbf{x}: g_1(x) = 1\}$ be the set of elements sampled by our algorithm at the end of the stream. Let $m_S = \sum_{i=1}^m \1[g_1(x_i) = 1]$ be the mass of elements sampled by our algorithm. Let $F_i = |\{x \in S: |\{j: x_j = x\}| = i\}|$ be the number of sampled elements with frequency $i$. Let $G$ be the number of nonempty buckets in the array $A$ at the end of the stream. We use $\log(x)$ to denote the logarithm of $x$ base $2$. \\

In the rest of this section, we start by bounding several quantities used by our algorithm. Then, the bulk of the analysis in \cref{subsec:inverting} focuses on the estimation error introduced by the post-processing procedure in \cref{alg:dp}.
To begin, we analyze our algorithm under the assumption that all of our hash functions are fully random and at the end in \cref{subsec:pseudorandom} show that our analysis only required limited randomness.

By the guarantees of \cite{blasiok2020distinctelements} and \cite{kane2010optimal}, the estimate of $\tilde{D}_t$ is correct up to relative error $1/10$ at all points in the stream and the estimate of $\hat{D}$ is correct up to relative error $\eps/10$ at the end of the stream with small constant failure probability and using independent randomness from the rest of the algorithm. In what follows, we will condition on the success of these estimators. 

\begin{lemma}\label{lem:space-exp-var}
At any point $t$ in the stream, the amount of space (in bits) required to store the array $A$ has expectation
\begin{equation*}
    O\left(B + 2^{-\ell_t} D_t \log(\min\{t/D_t, \tau\})\right)
\end{equation*}
and variance
\begin{equation*}
    O\left(2^{-\ell_t} D_t \log^2(\min\{t/D_t, \tau\})\right).
\end{equation*}
\end{lemma}

\begin{proof}
The variable-bit-length array construction given in~\cite{blandford2008CompactDF} guarantees that the space used to store the array $A$ is, up to constant factors, the number of buckets $B$ plus the total number of bits required to represent the data stored in the buckets. To prove the lemma, we will bound the latter quantity.

Consider a fixed point $t$ in the stream. Let $X_{i, \ell}$ be a random variable indicating the bit complexity of storing the level and count of element $i$ under the random event that $g_1(i)$ has least significant bit $\ell$. Let $f_t(i)$ be the number of times $i$ appears up to time $t$, truncated to be at most $\tau$. For $\ell < \ell_t$, $X_{i, \ell} = 0$ and otherwise if $f_t(i) > 0$, $X_{i,\ell} \leq \lceil \log(\ell - \ell_t + 1) \rceil + \lceil \log(f_t(i)) \rceil + 2$.
Let $Z_i = z(i)$ be the number of copies of element $i$ made for Poissonization. Recall that $Z_i \sim \poi(1)$.
Then, $\sum_{i=1}^n \sum_{\ell=1}^{\log n}  X_{i, \ell} Z_i$ upper bounds the space required to store all of the (level, counter) pairs in $A$ as separately storing counters for each sampled item rather than for each nonempty bucket will only use more space.

The expectation of the sum is
\begin{align*}
    \E\left[\sum_{i=1}^n \sum_{\ell=1}^{\log n} X_{i, \ell} Z_i\right] &= \sum_{\ell = \ell_t}^{\log n} \sum_{i \in [n]: f_t(i) > 0} 2^{-\ell} \left(\lceil \log(\ell - \ell_t + 1) \rceil + \lceil \log(f_t(i)) \rceil + 2\right) \\
    &= 2^{-\ell_t}D_t \sum_{\ell'=0}^{\log n - \ell_t} 2^{-\ell'}\left(2 + \lceil \log(\ell' + 1) \rceil\right) + \sum_{i \in [n]: f_t(i) > 0} \lceil \log(f_t(i)) \rceil \sum_{\ell = \ell_t}^{\log n} 2^{-\ell}.
\end{align*}
The first term above converges to $O(2^{-\ell_t} D_t)$ and the second term is bounded by 
\begin{equation*}
     O\left(  2^{-\ell_t} \sum_{i \in [n]: f_t(i) > 0}  \log(f_t(i)) \right). 
\end{equation*}
To analyze the second term, note that $\sum_{i \in [n]: f_t(i) > 0} f_t(i) \leq \min\{t, \tau D_t\}$. By Jensen's inequality and the concavity of logarithm,
\begin{equation*}
    D_t \sum_{i \in [n]: f_t(i) > 0}  \frac{1}{D_t} \log(f_t(i)) \leq D_t \log\left(\frac{1}{D_t} \sum_{i \in [n]: f_t(i) > 0} f_t(i) \right) = D_t \log(\min\{t, \tau D_t\}/D_t).
\end{equation*}
So, in total, the expected space required to store all (level, counter) pairs is
\begin{equation*}
    \E\left[\sum_{i=1}^n \sum_{\ell=1}^{\log n} X_{i, \ell} Z_i\right] = O\left(2^{-\ell_t} D_t \log(\min\{t/D_t, \tau\})\right).
\end{equation*}

We will continue by bounding the variance of the required space. For a fixed element $i$,
\begin{equation*}
    \Var\left(\sum_{\ell=1}^{\log n} X_{i, \ell} Z_i \right) \leq \sum_{\ell=1}^{\log n} \Var(X_{i, \ell} Z_i)
\end{equation*}
as $X_{i, \ell} Z_i$ are negatively correlated across different $\ell$: if $X_{i, \ell}$ is nonzero, all $X_{i, \ell'}$ for $\ell' \neq \ell$ are zero. In addition, since the $X_{i, \ell} Z_i$ variables are independent across different $i$,
\begin{equation*}
    \Var\left( \sum_{i=1}^n \sum_{\ell=1}^{\log n} X_{i, \ell} Z_i \right)
    \leq \sum_{i=1}^n  \sum_{\ell=1}^{\log n} \Var(X_{i, \ell} Z_i).
\end{equation*}
As $X_{i, \ell}$ and $Z_i$ use independent randomness and $\Var(Z_i)=1$,
\begin{align*}
    \Var(X_{i, \ell} Z_i) &= \Var(X_{i, \ell}) \Var(Z_i) + Var(X_{i, \ell}) \E[Z_i^2] + \E(X_{i, \ell}^2) \Var(Z_i) \\
    &\leq 4\E[X^2_{i, \ell}] \\
    &= 4 2^{-\ell} \left(\lceil \log(\ell - \ell_t + 1) \rceil + \lceil \log(f_t(i)) \rceil + 2\right)^2 \\
    &= O\left(2^{-\ell} \left(\log^2(\ell - \ell_t + 1) + \log^2(f_t(i)) + 1\right)\right)
\end{align*}
By the same calculation as for the expectation earlier,
\begin{equation*}
    \sum_{i=1}^n \sum_{\ell=1}^{\log n} 2^{-\ell} \left(\log^2(\ell - \ell_t + 1) + \log^2(f_t(i)) + 1\right)
    = O\left(2^{-\ell_t} D_t + 2^{-\ell_t} \sum_{i \in [n]: f_t(i) > 0} \log^2(f_t(i))\right).
\end{equation*}
As $\log^2(x + e)$ is concave for $x > 0$, 
\begin{equation*}
     \sum_{i \in [n]: f_t(i) > 0} \log^2(f_t(i)) \leq D_t \log^2(\min\{t, \tau D_t\}/D_t + e/D_t)
\end{equation*}
by Jensen's inequality.
So, in total, the variance is bounded by
\begin{equation*}
    \Var\left(\sum_{i=1}^n \sum_{\ell=1}^{\log n} X_{i, \ell} Z_i\right) = O\left(2^{-\ell_t} D_t \log^2(\min\{t/D_t, \tau\})\right).
\end{equation*}
\end{proof}

\begin{lemma}\label{lem:counts}
With constant probability, over all points in the stream $t \in [m]$, it holds that the space used to store $A$ is $O(B)$.
\end{lemma}

\begin{proof}
For any consecutive period of time during which $\ell_t$ remains unchanged, the space used by the algorithm never decreases. In addition, the points at which $\ell_t$ changes relies only on the stream length and tracking sketch and thus is independent of the other randomness used by the algorithm. Let $t_i$ be the last point in the stream with $\ell_{t_i} = i$ for $i \in \{1, \ldots, r\}$. It must be the case that the maximum space usage of the stream occurs at one of these points. Note that $r < O(\log n)$ as we never increase $\ell_t$ if $2^{\ell_t} = \Omega(n)$ (note that this is maintained implicitly for $errortype=D$ as $\tilde{D}_t \leq 1.1 n$). 

Let $Y$ be a random variable corresponding to the space used for storing $A$.
By \cref{lem:space-exp-var}, the expected space is
\begin{equation*}
    \E[Y] = O\left(B + 2^{-\ell_t} D_t \log(\min\{t/D_t, \tau\})\right).
\end{equation*}
And by Chebyshev's inequality,
\begin{equation*}
    \Pr\left(|Y - \E[Y]| \geq \alpha B\right) = O\left(\frac{B + 2^{-\ell_t} D_t \log^2(\min\{t/D_t, \tau\})}{\alpha^2 B^2}\right)
\end{equation*}
for any $\alpha > 0$.
We will proceed separately for each error type.

\begin{alphaenumerate}
    \item $errortype=m$:
    By the condition on Line~10 
    in \cref{alg:profile}, we always maintain that $2^{\ell_t}$ is at least  $\min\{t/KB, n\}$ for constant $K$. 
    Noting that $t \geq D_t$ and $B = \Omega(\log n)$, the expected space is
    \begin{equation*}
        O\left(B + \frac{B D_t \log(t/D_t)}{t} + \frac{D_t\log(t/D_t)}{n}\right) = O(B).
    \end{equation*}
    If $2^{\ell_t} = \Theta(t/B)$, the probability of deviation is bounded by 
    \begin{equation*}
        \Pr\left(|Y - \E[Y]| \geq \alpha B\right) = O\left(\frac{1}{\alpha^2 B} + \frac{D_t \log^2(t/D_t)}{\alpha^2 t B}\right) = O\left(\frac{1}{\alpha^2 B}\right).
    \end{equation*}
    If $2^{\ell_t} \geq 2n$,
    \begin{equation*}
        \Pr\left(|Y - \E[Y]| \geq \alpha B\right) = O\left(\frac{1}{\alpha^2 B} + \frac{D_t\log^2(t/D_t)}{n \alpha^2 B^2}\right) = O\left(\frac{1}{\alpha^2}\right).
    \end{equation*}
    Once $2^{\ell_t}$ exceeds $2n$, it will never be incremented again, so for $t \in \{t_1, \ldots, t_{r-1}\}$, it must be the case that $2^{\ell_t} < 2n$.
    In addition, for these timesteps, it must be the case that $2^{\ell_t} \approx t/KB$ as, by definition, after adding one additional point to the stream, $2^{\ell_t}$ will exceed $t/KB$.
    Union bounding over all $t_1, \ldots, t_r$, the probability that the space ever exceeds $C B$ for some large enough constant $C$ is at most
    \begin{align*}
        O\left(\frac{r-1}{\alpha^2 B} + \frac{1}{\alpha^2}\right) = O\left(\frac{1}{\alpha^2}\right).
    \end{align*}
    So, choosing some appropriate constant $\alpha$, the probability that the space ever exceeds $CB$ over the course of the stream is bounded by a small constant.

    \item $errortype=D$ and $\tau = O(1)$:
    By the condition on Line~\ref{line:update-level} in \cref{alg:profile}, we always maintain that $2^{\ell_t}$ is at least $\tilde{D}_t/BK$ for constant $K$. Therefore, $2^{-\ell_t} D_t \log \tau = O(B)$ and the expected space and variance is $O(B)$.
    For $\alpha > 0$, the probability of deviation is bounded by 
    \begin{equation*}
        \Pr\left(|Y - \E[Y]| \geq \alpha B\right)
        = O\left(\frac{1}{\alpha^2 B}\right).
    \end{equation*}
    Union bounding over all $t_1, \ldots, t_r$ and using the assumption that $B = \Omega(\log n)$, the probability that the space ever exceeds $\alpha B$ is at most
    \begin{align*}
        O\left(\frac{r}{B\alpha^2}\right) = O\left(\frac{1}{\alpha^2}\right).
    \end{align*}
    So, choosing some appropriate constant $\alpha$, the probability that the space ever exceeds $\alpha B$ over the course of the stream is bounded by a small constant.
\end{alphaenumerate}
\end{proof}

\begin{lemma}\label{lem:S} 
With constant probability, the size of the final sample is bounded as follows:
\begin{alphaenumerate}
    \item If $errortype=m$, $|S| = \Theta\left(\max\{DB/m, D/n\}\right)$.
    \item If $errortype=D$, $|S| = \Theta(B)$.
\end{alphaenumerate}
\end{lemma}
\begin{proof}
Conditioning on the final level $\ell_m$, each element is sampled with probability $2^{-\ell_m}$. $|S|$ is distributed as a Binomial random variable with expectation $2^{-\ell_m} D$ as there are $D$ possible elements to sample.

For  $errortype = m$, after a certain point in the stream, we maintain that $2^{\ell_t} \approx \min\{tK/B, n\}$ (up to a factor of two), so $2^{-\ell_m} D = \Theta(\max\{D B/m, D/n\})$. 
While $D_t$ can be significantly smaller than $D_t B / t$ early on, since $B = O(D)$, by the end of the stream we will have reached a point where $D_t = \Omega(D_t B/t)$.

For $errortype = D$, after a certain point in the stream, we maintain that $2^{\ell_t} \approx \tilde{D}_t K/B$ (up to a factor of two), so $2^{-\ell_m} D = \Theta(B)$. While $2^{-\ell_t} D_t$ can be significantly smaller than $B$ early on, since $B = O(D)$, by the end of the stream they must be the same order of magnitude.

As the variance of a Binomial random variable is at most its expectation, the results follow from Chebyshev's inequality.
\end{proof}


\begin{lemma}\label{lem:poisson}
With constant probability, the number of elements with frequency $i$ in a given bucket is distributed as $\poi(F_i/B)$.
\end{lemma}
\begin{proof}
Let $x\in S$ be a sampled element with frequency $i$.
Each time $x$ appears in the stream, $Z_x \sim \poi(1)$ copies of it are added to the stream.
Summing over all such $x \in S$, the total number of (copied) elements with frequency $i$ added to $A$ is  distributed as $\poi(F_i)$.
By the ``Poissonization'' trick, the number of elements with frequency $i$ in a given bucket is distributed as $\poi(F_i/B)$, as desired.

The one caveat is that we truncate $Z_x$ to be at most $H$. However, the probability that a $\poi(1)$ random variable exceeds $H$ is at most $e^{-O(H)}$. As $S = O(B)$, setting $H = \Theta(\log B)$ suffices for $Z_x$ to not exceed $H$ with constant probability after union bounding over all $x \in S$.
\end{proof}

Recall that $\hat{S} = - B \ln\left(1 - \frac{G}{B}\right)$ is our estimate of $|S|$.
\begin{lemma}\label{lem:hatS}
With constant probability, $|\hat{S} - |S|| = O(\sqrt{|S|})$.
\end{lemma}
\begin{proof}
The number of samples within a given bin is distributed as a $\poi(|S|/B)$ random variable and is independent across bins.
Therefore, the probability that a given bin is nonempty is $1 - e^{-|S|/B}$.
The number of nonempty bins $G$ is thus distributed as a $\text{Binomial}(B, 1 - e^{-|S|/B})$ random variable with expectation
\begin{equation*}
    \E[G] =  B \left(1 - e^{-|S|/B}\right)
\end{equation*}
Let $\delta_G = G - \E[G]$.
The variance of $G$ is at most its expectation, so with constant probability, its absolute deviation from its expectation is at most $|\delta_G| = O\left(\sqrt{B \left(1 - e^{-|S|/B}\right)}\right)$.
Now, consider $\hat{S}$ under this assumption on $\delta_G$:
\begin{align*}
    ||S| - \hat{S}| &= \left||S| + B\ln\left(1 - \frac{G}{B}\right)\right| \\
    &= \left| |S| +B\ln\left(1 - \frac{B \left(1 - e^{-|S|/B}\right) + \delta_G}{B}\right) \right|\\
    &= \left||S| + B \ln \left(e^{-|S|/B} + \frac{\delta_G}{B}\right)\right| \\
    &= \left||S| + B \left(\ln \left(e^{-|S|/B}\right) + \ln\left(1 + \frac{\delta_G}{Be^{-|S|/B}}\right)\right) \right|\\
    &= \left||S| - |S| + B\ln\left(1 + \frac{\delta_G}{Be^{-|S|/B}}\right)\right| \\
    &\leq \left|\frac{B\delta_G}{B e^{-|S|/B}}\right| \\
    &= O\left( \frac{\sqrt{B(1 - e^{-|S|/B})}}{e^{-|S|/B}}\right) \\
    &= O\left(\sqrt{B \left(e^{|S|/B} - 1\right)}\right) \\
    &= O(\sqrt{|S|}).
\end{align*}
\end{proof}

Before we continue, it will be useful to have the following lemmas.
\begin{lemma}\label{lem:sumofsqrt}
Given a sequence of non-negative integers $x_1, \ldots, x_\tau$ such that $\sum_{i=1}^\tau x_i \cdot i \leq m$ and $\sum_{i=1}^\tau x_i \leq D$,
\begin{equation*}
    \sum_{i=1}^\tau \sqrt{x_i} = O(\sqrt{m \log \tau})
\end{equation*}
and
\begin{equation*}
    \sum_{i=1}^\tau \sqrt{x_i} = O(\sqrt{D\tau})
\end{equation*}
\end{lemma}

\begin{proof}
Consider the $\tau$-dimensional vectors $a, b$ such that $a_i = \sqrt{x_i \cdot i}$ and $b_i = 1/\sqrt{i}$.
Note that $\langle a, b \rangle = \sum_{i=1}^\tau \sqrt{x_i}$,
$\|a\|_2 = \sqrt{\sum_{i=1}^\tau x_i \cdot i} \leq \sqrt{m}$, and
$\|b\|_2 = \sqrt{\sum_{i=1}^\tau 1/i} = O(\sqrt{\log \tau})$.
The first result of lemma then follows from Cauchy-Schwarz.

For the second result, consider vectors $a, b$ such that $a_i = \sqrt{x_i}$ and $b_i = 1$.
Then, $\langle a, b \rangle = \sum_{i=1}^\tau \sqrt{x_i}$,
$\|a\|_2 = \sqrt{D}$, and
$\|b\|_2 = \sqrt{\tau}$.
The second result of lemma then follows from Cauchy-Schwarz.
\end{proof}

We will now prove that the $L_1$ profile estimation error is small if we rescale the empirical profile of the sampled elements $F_1, \ldots, F_{\tau}$.
Recall that in the algorithm we only estimate this empirical profile (the bulk of the later analysis will focus on the quality of that estimation).
\begin{lemma}\label{lem:empiricalprofile1}
\begin{equation*}
    \E\left[\sum_{i=1}^{\tau} \left| \phi_i - \frac{D}{|S|} F_i \right|\right] \leq \sqrt{\frac{D}{|S|}} \sum_{i=1}^{\tau} \sqrt{\phi_i}.
\end{equation*}
\end{lemma}

\begin{proof}
Note that $F_i \sim Bin(|S|, \frac{\phi_i}{D})$. We will analyze the expectation and variance of our estimate $\hat\phi = \frac{D}{|S|}F_i$:
\begin{equation*}
    \E[\hat\phi] = \frac{D}{|S|} \E[F_i] = \phi_i
\end{equation*}
and
\begin{equation*}
    \Var(\hat\phi) = \frac{D^2}{|S|^2} \Var(F_i) \leq \frac{D \phi_i}{|S|}.
\end{equation*}
Let $\delta_i = \left|\phi_i - \frac{D}{|S|}F_i\right|$ be the error in the $i$th coordinate.
The expectation of $\delta_i$ is
\begin{equation*}\label{eq:experrorcoord}
    \E\left[\delta_i\right] = \E\left[\sqrt{(\phi_i - \hat{\phi_i})^2}\right] \leq \sqrt{\E\left[\left(\phi_i - \hat{\phi_i}\right)^2\right]} = \sqrt{\Var(\hat{\phi_i})} = \sqrt{\frac{D\phi_i}{|S|}}.
\end{equation*}
(Using Jensen's inequality applied to the square root).
By linearity of expectation, the total expected error is
\begin{equation*}
    \E\left[\sum_{i=1}^{\tau} \delta_i\right] = \sum_{i=1}^{\tau} \E[\delta_i] \leq \sqrt{\frac{D}{|S|}} \sum_{i=1}^{\tau} \sqrt{\phi_i}.
\end{equation*}
\end{proof}

\begin{lemma}\label{lem:empiricalprofile2}
For $errortype=m$ if $B = \Theta\left(\frac{\log \tau}{\eps^2}\right)$,
\begin{equation*}
    \sum_{i=1}^{\tau} \left| \phi_i - \frac{D}{|S|} F_i \right| = O(\eps m),
\end{equation*}
and for $errortype=D$  and $\tau = O(1)$ if $B = \Theta\left(\frac{1}{\eps^2}\right)$,
\begin{equation*}
    \sum_{i=1}^{\tau} \left| \phi_i - \frac{D}{|S|} F_i \right| = O(\eps D),
\end{equation*}
both holding with constant probability.
\end{lemma}

\begin{proof}
As $\sum_{i=1}^\tau \phi_i \cdot i = m$ and $\sum_{i=1}^\tau \phi_i = D$, \cref{lem:sumofsqrt} implies the following bounds
\begin{equation*}
    \sum_{i=1}^{\tau} \sqrt{\phi_i} = O(\sqrt{m \log \tau})
\end{equation*}
and
\begin{equation*}
    \sum_{i=1}^{\tau} \sqrt{\phi_i} = O(\sqrt{D}).
\end{equation*}

Consider $errortype=m$. By \cref{lem:S}, $S = \Omega(DB/m)$ for $errortype=m$. By \cref{lem:empiricalprofile1}, the expected error of the empirical estimate is bounded by
\begin{equation*}
    \E\left[\sum_{i=1}^{\tau} \left| \phi_i - \frac{D}{|S|} F_i \right|\right] = O\left(\sqrt{\frac{D m \log \tau}{|S|}}\right) = O\left(\sqrt{\frac{m^2 \log(\tau)}{B}}\right) = O(\eps m).
\end{equation*}

Now consider $errortype=D$. By \cref{lem:S}, $S = \Omega(D)$ for $errortype=D$. By \cref{lem:empiricalprofile1}, the expected error of the empirical estimate is bounded by
\begin{equation*}
    \E\left[\sum_{i=1}^{\tau} \left| \phi_i - \frac{D}{|S|} F_i \right|\right] = O\left(\sqrt{\frac{D^2}{|S|}}\right) = O\left(D \sqrt{\frac{1}{B}}\right) = O(\eps D).
\end{equation*}
Markov's inequality completes the proof.
\end{proof}

\subsection{Inverting Counts}\label{subsec:inverting}
Let $b_i$ be the number of buckets with count exactly $i$.
Let $r_i$ be the number of buckets with count $i$ that are formed by collisions of elements with smaller frequencies.
Let $s_i$ be the number of buckets with count $i$ that are formed by a single element with frequency $i$ falling in that bucket and no other elements falling in that bucket. Note that $b_i = r_i + s_i$.

The core idea of our post-processing procedure is to recursively estimate $F_i$ for $i=1,\ldots, \tau$ by relating the expectation of $F_i$ to the expectation of $s_i$. Using prior estimates $\hat{F}_1, \ldots, \hat{F}_{i-1}$, we will approximate $\E[r_i]$. Then, using $b_i$ as an estimate of $\E[b_i]$, we will plug in all of these estimates to solve for $F_i$. Along the way, errors will be introduced both due to random deviations as well as error which propagate from the fact that we do not know the true $F_j$'s for $j < i$. The core technical challenge of our analysis is to bound these errors.

Let $X_{i, k} \sim \poi(F_i/B)$ be a random variable corresponding to the number of elements with count $i$ in bucket $k$. By Poissonization, $X_{i, k}$ is mutually independent of all $X_{i', k'}$ for $(i, k) \neq (i', k')$.
Let $X_k = \sum_{i=1}^m X_{i, k} \sim \poi(|S|/B)$ be the random variable associated with the number of elements in bucket $k$.
\begin{align*}
    \E[s_i] &= \sum_{k=1}^B \Pr(\text{bucket $k$ contains a unique element which has count $i$}) \\
    &= \sum_{k=1}^B \Pr(X_{i, k} = 1) \prod_{j \in [m]: j \neq i} \Pr(X_{j, k} = 0) \\
    &= \sum_{k=1}^B \frac{F_i}{B} \left(e^{-F_i/B}\right)\prod_{j \in [m]: j \neq i} e^{-F_j/B} \\
    &= B \left(\frac{F_i}{B}\right) e^{-|S|/B} \\
    &= F_i e^{-|S|/B}
\end{align*}

As $s_i = b_i - r_i$, we can express $F_i$ as
\begin{equation}\label{e:F-br}
    F_i = \E[b_i e^{|S|/B}] - \E[r_i e^{|S|/B}].
\end{equation}
As the expectations of the $b_i$'s and $r_i$'s depend on the true values $F_1, \ldots, F_i$, we cannot calculate them exactly. Rather, our estimate $\hat{F}_i$ will be formed by plugging into \cref{e:F-br} the empirical count of $b_i$ and the approximation $\hat{r}_i$ of $\E[r_i]$ formed using our previous estimates $\hat{F}_1, \ldots, \hat{F}_{i-1}$ in place of the true $F_j$'s. Further, as we do not know $|S|$, we will need to plug in an estimate $\hat{S}$ wherever it appears.

We will now show that we can express $\E[r_i]$ as a function of $F_1, \ldots F_{i-1}$. Let $Y_i \sim \poi(F_i/B)$ be a random variable corresponding to the number of elements with count $i$ in a given bucket. Let $\mathbf{y}= y_1, \ldots, y_m$ denote a vector corresponding to a specific assignment of how many distinct elements of counts $1, \ldots, m$ appear in a given bucket.
\begin{align*}
    \E[r_i] &= B \Pr(\text{a bucket has summed count $i$ and at least two distinct elements})\\
    &= B \sum_{\mathbf{y}: (\sum_{j=1}^m y_j \geq 2) \land  (\sum_{j=1}^m y_j\cdot j = i)} \prod_{j=1}^m \Pr(Y_j = y_j) \\
    &= B  \sum_{\mathbf{y}: (\sum_{j=1}^m y_j \geq 2) \land  (\sum_{j=1}^m y_j\cdot j = i)} \prod_{j=1}^m \left(\frac{F_j}{B}\right)^{y_j} \frac{e^{-F_j/B}}{y_j!} \\
    &= B e^{-|S|/B}  \sum_{\mathbf{y}: (\sum_{j=1}^m y_j \geq 2) \land  (\sum_{j=1}^m y_j\cdot j = i)} \prod_{j=1}^{i-1} \left(\frac{F_j}{B}\right)^{y_j} \frac{1}{y_j!}
\end{align*}
Let $\hat{r_i}(\hat{F}_1, \ldots \hat{F}_{i-1})$ be the quantity we get by calculating $\E[r_i e^{|S|/B}]$ under estimated parameters $\hat{F}_1, \ldots \hat{F}_{i-1}$:
\begin{equation} \label{e:hatr}
    \hat{r_i}(\hat{F}_1, \ldots \hat{F}_{i-1})
    = B  \sum_{\mathbf{y}: (\sum_{j=1}^{i-1} y_j \geq 2) \land  (\sum_{j=1}^{i-1} y_j\cdot j = i)} \prod_{j=1}^{i-1} \left(\frac{\hat{F}_j}{B}\right)^{y_j} \frac{1}{y_j!}
\end{equation}
Then, our final estimate for $\hat{F_i}$ will be:
\begin{equation}
    \hat{F}_i = \max\{b_i e^{\hat{S}/B} - \hat{r}_i(\hat{F}_1, \ldots, \hat{F}_{i-1}), 0\}.
\end{equation}

\subsubsection{Dynamic Programming}
Calculating these estimates na\"ively requires $\text{exp}(\tau)$ time as $\hat{r}_i(\cdot)$ in \cref{e:hatr} is a summation over integer partitions of $i$. Therefore, we use a dynamic program in \cref{alg:dp} to calculate this expression efficiently.

Given positive integral parameters $j, x$, the quantity of interest will be the expected number of buckets under $\hat{F}_1, \ldots, \hat{F}_{i-1}$ which have total (summed) count $j$ and minimum frequency element with frequency $x$.
We will use the following description of the dynamic program which is expressed equivalently in pseudocode in \cref{alg:dp}:
\begin{equation}\label{e:dp}
    \DP[j,x] = \begin{cases}
        \hat{F}_j \; &\text{ if $x = j$ and $j < i$}\\
        \sum_{k=1}^{\lfloor j/x \rfloor - 1} \sum_{x' = x+1}^{j - kx} \DP[j - kx, x'] \left(\frac{\hat{F}_x}{B}\right)^k \frac{1}{k!}
        &\text{ if $x \leq \lfloor j/2 \rfloor$ and $j \neq 0 \bmod{x}$} \\
        \sum_{k=1}^{\lfloor j/x \rfloor - 1} \sum_{x' = x+1}^{j - kx} \DP[j - kx, x'] \left(\frac{\hat{F}_x}{B}\right)^k \frac{1}{k!}
        + \left(\frac{\hat{F}_x^{j/x}}{B^{j/x - 1}}\right) \frac{1}{(j/x)!}
        &\text{ if $x \leq \lfloor j/2 \rfloor$ and $j = 0 \bmod{x}$}\\
        0 \; &\text{ otherwise}
    \end{cases}
\end{equation}

\begin{lemma}
$\hat{r}_i(\hat{F}_1, \ldots, \hat{F}_{i-1}) = \sum_{x = 1}^{\lfloor i/2 \rfloor} \DP[i,x]$.
\end{lemma}

\begin{proof}
Consider the following quantity which restricts $\hat{r}_i$ to occurrences with minimum frequency $x$:
\begin{equation}
    \hat{r}_{i, x}(\hat{F}_1, \ldots, \hat{F}_{i-1})
    = B  \sum_{\mathbf{y}: (\sum_{j=1}^{i-1} y_j \geq 2) \land  (\sum_{j=1}^{i-1} y_j\cdot j = i) \land (\sum_{j=1}^{x} y_j = y_x \geq 1)} \prod_{j=1}^{i-1} \left(\frac{\hat{F}_j}{B}\right)^{y_j} \frac{1}{y_j!}.
\end{equation}
Note that $\hat{r}_i(\hat{F}_1, \ldots, \hat{F}_{i-1}) = \sum_{x=1}^{\lfloor i/2 \rfloor} \hat{r}_{i,x}(\hat{F}_1, \ldots, \hat{F}_{i-1})$.

We will prove via induction on $i$ that $\DP[i,x] = \hat{r}_{i,x}(\hat{F}_1, \ldots, \hat{F}_{i-1})$ for all $x \in \{1,\ldots \lfloor i/2 \rfloor\}$.
For the base case, consider $i = 2$. We must show that $\DP[2,1] = \hat{r}_{2,1}(\hat{F}_1)$.
Considering the left hand side first, the sum over values of $k$ in the third case of \cref{e:dp} is empty, so
\begin{equation*}
    \DP[2,1] = \left(\frac{\hat{F}_1^2}{B}\right)\frac{1}{2}.
\end{equation*}
For the right hand side, the only valid assignment $\mathbf{y}$ is the one in which $y_1 = 2$ and all other variables are zero.
Therefore,
\begin{equation*}
    \hat{r}_{2, 1}(\hat{F}_1) = B \left(\frac{\hat{F}_1}{B}\right)^2 \frac{1}{2} = \left(\frac{\hat{F}_1^2}{B}\right)\frac{1}{2},
\end{equation*}
completing the base case.

For the inductive case, fix $i > 1$ and assume for all $j < i$ that $\DP[j, x] = \hat{r}_{j,x}(\hat{F}_1, \ldots, \hat{F}_{j-1})$ for $x \in \{1, \ldots, \lfloor j/2 \rfloor\}$.
Consider any $x \in \{1,\ldots, \lfloor i/2 \rfloor\}$ such that $i \neq 0 \bmod x$. We will address the special case where $i$ is a multiple of $x$ later.
The core idea is to decompose the sum over assignments with total count $i$ and minimum value $x$ into sums over assignments with total counts less than $i$ and minimum values greater than $x$:
\begin{align*}
    \hat{r}_{i, x}(\hat{F}_1, \ldots, \hat{F}_{i-1})
    &= B  \sum_{\mathbf{y}: (\sum_{j=1}^{i-1} y_j \geq 2) \land  (\sum_{j=1}^{i-1} y_j\cdot j = i) \land (\sum_{j=1}^{x} y_j = y_x \geq 1)} \prod_{j=1}^{i-1} \left(\frac{\hat{F}_j}{B}\right)^{y_j} \frac{1}{y_j!} \\
    &= B  \sum_{\mathbf{y}: (\sum_{j=1}^{i-1} y_j - y_x \geq 1) \land  (\sum_{j=1}^{i-1} y_j \cdot j - y_x \cdot x = i - y_x \cdot x) \land (\sum_{j=1}^{x} y_j = y_x \geq 1)} \prod_{j=1}^{i-1} \left(\frac{\hat{F}_j}{B}\right)^{y_j} \frac{1}{y_j!} \\
    &= B \sum_{k=1}^{\lfloor i/x \rfloor - 1} \left(\frac{\hat{F}_x}{B}\right)^{k} \frac{1}{k!}
    \sum_{\mathbf{y}: (\sum_{j=1}^{i-1} y_j \geq 1) \land  (\sum_{j=1}^{i-1} y_j \cdot j = i - kx) \land (\sum_{j=1}^{x} y_j = 0)} \prod_{j=1}^{i-1} \left(\frac{\hat{F}_j}{B}\right)^{y_j} \frac{1}{y_j!}.
\end{align*}
We can break down the sum over assignments in the final equation above into two parts. It contains one term which corresponds to the assignment with a single element of frequency $i - kx$ and no other elements. The rest of the terms correspond to the sum over all assignments with at least two elements whose sum is $i - kx$ and whose minimum frequency is greater than $x$. This latter sum is equivalent to $\frac{1}{B} \sum_{x'=x+1}^{\lfloor (i-kx)/2 \rfloor} \hat{r}_{i-kx, x'}(\hat{F}_1, \ldots, \hat{F}_{i-kx-1})$:
\begin{align*}
    \hat{r}_{i, x}(\hat{F}_1, \ldots, \hat{F}_{i-1})
    &= B \sum_{k=1}^{\lfloor i/x \rfloor - 1} \left(\frac{\hat{F}_x}{B}\right)^{k} \frac{1}{k!}
    \left(\frac{\hat{F}_{i - kx}}{B} + \frac{1}{B}\sum_{x' = x+1}^{\lfloor (i - kx)/2 \rfloor} \hat{r}_{i - kx, x'}(\hat{F}_1, \ldots, \hat{F}_{i - kx - 1})\right) \\
    &= \sum_{k=1}^{\lfloor i/x \rfloor - 1} \left(\frac{\hat{F}_x}{B}\right)^{k} \frac{1}{k!}
    \left(\hat{F}_{i - kx} + \sum_{x' = x+1}^{\lfloor (i - kx)/2 \rfloor} \hat{r}_{i - kx, x'}(\hat{F}_1, \ldots, \hat{F}_{i - kx - 1})\right).
\end{align*}
Recall that $\DP[j, j] = \hat{F}_j$. Combined with the inductive hypothesis,
\begin{align*}
    \hat{r}_{i, x}(\hat{F}_1, \ldots, \hat{F}_{i-1})
    &= \sum_{k=1}^{\lfloor i/x \rfloor - 1} \left(\frac{\hat{F}_x}{B}\right)^{k} \frac{1}{k!}
    \left(\DP[i - kx, i - kx] + \sum_{x' = x+1}^{\lfloor (i - kx)/2 \rfloor} \DP[i - kx, x']\right).
\end{align*}
As $\DP[j, x] = 0$ for $x \in \{\lfloor j/2 \rfloor + 1,  \ldots, j-1\}$, we can extend the range of the sum over values of $x'$ to $\{x+1, \ldots, i-kx-1\}$:
\begin{align*}
    \hat{r}_{i, x}(\hat{F}_1, \ldots, \hat{F}_{i-1})
    &= \sum_{k=1}^{\lfloor i/x \rfloor - 1} \left(\frac{\hat{F}_x}{B}\right)^{k} \frac{1}{k!}
    \left(\sum_{x' = x+1}^{i - kx} \DP[i - kx, x']\right) \\
    &= \sum_{k=1}^{\lfloor i/x \rfloor - 1} \sum_{x' = x+1}^{i - kx} \DP[i - kx, x'] \left(\frac{\hat{F}_x}{B}\right)^{k} \frac{1}{k!} \\
    &= \DP[i, x],
\end{align*}
completing the induction.

It remains to address the case where $i = 0 \bmod x$. In this case, $\hat{r}_{i, x}(\hat{F}_1, \ldots, \hat{F}_{i-1})$ has an additional term corresponding to an assignment $\mathbf{y}$ where $y_x = i/x$ and all other terms are zero. This corresponds to an additive term
\begin{equation*}
        B \left(\frac{\hat{F}_x}{B}\right)^{i/x} \frac{1}{(i/x)!} = \left(\frac{\hat{F}_x^{i/x}}{B^{i/x - 1}}\right) \frac{1}{(i/x)!}
\end{equation*}
which is accounted for exactly by the additional term in the dynamic program \cref{e:dp} in the case where $i$ is a multiple of $x$.

\end{proof}

\subsubsection{Error Analysis}
Let $\zeta_i$ be a random variable for the error introduced by using our old estimates to evaluate $\hat{r_i}$:
\begin{equation}
    \zeta_i = |\hat{r_i}(F_1, \ldots, F_{i-1}) - \hat{r_i}(\hat{F}_1, \ldots, \hat{F}_{i-1})|.
\end{equation}
Let $\gamma$ be a random variable for the rest of the error in estimating $F_i$ due to randomness deviations in $b_i$ as well as due to our error in approximating $|S|$ as $\hat{S}$,
\begin{equation}
    \gamma_i = |\E[b_i]e^{|S|/B} - b_ie^{\hat{S}/B}|.
\end{equation}
Note that as $F_i \geq 0$, thresholding our estimate $\hat{F}_i$ to always be at least zero only reduces the error:
\begin{equation*}
    |F_i - \hat{F}_i| \leq \gamma_i + \zeta_i.
\end{equation*}

\begin{lemma}\label{lem:errprop}
$\sum_{i=1}^{\tau} \zeta_i \leq \sum_{i=1}^{\tau} \gamma_i$.
\end{lemma}

To analyze the $\zeta_i$ terms, we will introduce a new quantity $f_j^i(\Delta)$ which measures the sensitivity of $\hat{r_i}$ to a change in the $j$th coordinate for $j \in \{1, \ldots, i-1\}$ of magnitude $\Delta \geq 0$.
Let $\mathcal{F} = \{F'_1, \ldots, F'_{i-1} \in \N \cup \{0\}: \sum_j F'_j \leq C_1 |S|\}$. Then,
\begin{equation} \label{e:sensitivity}
    f_j^i(\Delta) = \sup_{F'_1, \ldots, F'_{i-1} \in \mathcal{F}, z \in \{\pm 1\}} \left| \hat{r_i}(F'_1, \ldots, F'_j, \ldots, F'_{i-1}) - \hat{r_i}(F'_1, \ldots, \max\{F'_j + z \Delta, 0\}, \ldots, F'_{i-1}) \right|
\end{equation}
First, for appropriate $C_1$, with constant probability, we will prove $\sum_{i=1}^{\tau} \hat{F}_i \leq C_1 |S|$ so that we can relate these quantities to the $\zeta_i$ terms. Note that for any valid input, $\hat{r}_i(F'_1, \ldots, F'_{i-1}) \geq 0$.
Therefore,
\begin{equation*}
    0 \leq \hat{F}_i \leq b_i e^{\hat{S}/B}.
\end{equation*}
Note that the sum over $b_i$'s is exactly the number of nonempty bins: $\sum_{i=1}^\tau b_i \leq \sum_{i=1}^m b_i = G < \hat{S}$. By the concentration of those quantities in \cref{lem:hatS}, with constant probability, $\sum_{i=1}^\tau \hat{F}_i \leq (C_1/e) |S| e^{\hat{S}/B} \leq C_1 |S|$ for large enough $C_1$.

We will prove \cref{lem:errprop} via a sequences of intermediate lemmas relating the $\zeta_i$ terms to the sensitivity of $\hat{r}_i(\cdot)$ as captured by $f_j^i(\Delta)$ and then bounding these sensitivities.

\begin{lemma}\label{lem:zeta-f}
\begin{equation*}
    \zeta_i \leq \sum_{k=1}^i \sum_{1 \leq j_1 < j_2 < \ldots < j_k \leq i-1} f_{j_{k}}^i \left( f_{j_{k-1}}^{j_k} \left( \ldots \left( f_{j_2}^{j_3} \left( f_{j_1}^{j_2} \left(\gamma_{j_1}\right) \right) \right)\right)\right).
\end{equation*}
\end{lemma}
\begin{proof}
Intuitively, each term in this expression accounts for the error that propagates from $\hat{F}_{j_1}$ to $\hat{F}_{j_2} \ldots$ to $\hat{F}_{j_k}$ and ultimately to $\hat{F}_i$.
We will prove the claim inductively.
For the base case, consider $i = 2$.
Recall that $\zeta_1 = 0$.
Then,
\begin{align*}
    \zeta_2 &= |\hat{r}_2(F_1) - \hat{r}_2(\hat{F}_1)| \\
    &= |\hat{r}_2(F_1) - \hat{r}_2(F_1 \pm \gamma_1)| \\
    &\leq f_1^2(\gamma_1).
\end{align*}
For the inductive step, consider some $i > 2$ and assume the claim holds for all $i' < i$:
\begin{align*}
    \zeta_i &= |\hat{r}_i(F_1, \ldots, F_{i-1}) - \hat{r}_i(\hat{F}_1, \ldots, \hat{F}_{i-1})| \\
    &=|\hat{r}_i(F_1, \ldots, F_{i-1}) - \hat{r}_i(\max\{F_1 \pm \gamma_1 \pm \zeta_1, 0\}, \ldots, \max\{F_{i-1} \pm \gamma_{i-1} \pm \zeta_{i-1}), 0\}| \\
    &\leq \sum_{j=1}^{i-1} f_j^i(\gamma_j) + f_j^i(\zeta_i) \\
    &\leq \sum_{j=1}^{i-1} \left(f_j^i(\gamma_j) + f_j^i\left(\sum_{k=1}^j \sum_{1 \leq j_1 < j_2 < \ldots < j_k \leq j-1} f_{j_{k-1}}^{j_k} \left( \ldots \left( f_{j_2}^{j_3} \left( f_{j_1}^{j_2} \left(\gamma_{j_1}\right) \right) \right)\right)\right) \right).
\end{align*}
The last inequality comes from the inductive hypothesis. Note that as $f_j^i(a)$ is the supremum with respect to all valid assignments of $F'_1, \ldots, F'_{i-1}$, $f_j^i(a + b) \leq f_j^i(a) + f_j^i(b)$. Therefore,
\begin{align*}
    \zeta_i &\leq \sum_{j=1}^{i-1} \left( f_j^i(\gamma_j) + \sum_{k=1}^j \sum_{1 \leq j_1 < j_2 < \ldots < j_k \leq j-1} f_j^i\left( f_{j_{k-1}}^{j_k} \left( \ldots \left( f_{j_2}^{j_3} \left( f_{j_1}^{j_2} \left(\gamma_{j_1}\right) \right) \right)\right)\right)\right) \\
    &= \sum_{j=1}^{i-1} \sum_{k=1}^j \sum_{1 \leq j_1 < j_2 < \ldots < j_k = j} f_{j_k}^i\left( f_{j_{k-1}}^{j_k} \left( \ldots \left( f_{j_2}^{j_3} \left( f_{j_1}^{j_2} \left(\gamma_{j_1}\right) \right) \right)\right)\right) \\
    &= \sum_{k=1}^{i-1} \sum_{1 \leq j_1 < j_2 < \ldots < j_k \leq i-1} f_{j_k}^i\left( f_{j_{k-1}}^{j_k} \left( \ldots \left( f_{j_2}^{j_3} \left( f_{j_1}^{j_2} \left(\gamma_{j_1}\right) \right) \right)\right)\right). 
\end{align*}
This completes the induction, proving the claim.
\end{proof}

Now, we will bound the quantities $f_j^i(\Delta)$ to get a bound on the $\zeta_i$ terms.
Fix $j$ and consider the quantity $\sum_{i=j+1}^{\tau} f_j^i(\Delta)$: the summed sensitivities of $\hat{r}_{j+1}(\cdot), \ldots, \hat{r}_{\tau}(\cdot)$ to changing $F'_j$ by $\Delta$ for $\Delta \leq C_1 |S|$. We will show that the total effect on $\hat{r}_i$ for $i > j$ from changing the estimate of $F'_j$ is the magnitude of the change.
\begin{lemma}\label{lem:sensitivity}
\begin{equation*}
    \sum_{i=j+1}^{\tau} f_j^i(\Delta) \leq \frac{\Delta}{2}.
\end{equation*}
\end{lemma}

\begin{proof}
Let $z_j, \Delta_j$ be chosen to achieve the maximum value in \cref{e:sensitivity}:
\begin{equation*}
    f_j^i(\Delta) =  \left| \hat{r_i}(F'_1, \ldots, F'_j, \ldots, F'_{i-1}) - \hat{r_i}(F'_1, \ldots, F'_j + z_j \Delta_j, \ldots, F'_{i-1}) \right|
\end{equation*}
with $\Delta_j = \Delta$ if $F'_j + z_j \Delta_j > 0$ and $\Delta_j = F'_j$ otherwise.
Expanding the definition of $\hat{r}_i(\cdot)$ in \cref{e:hatr},
\begin{align*}
    \sum_{i=j+1}^{\tau} f_j^i(\Delta)
    =& \sum_{i=j+1}^{\tau}
    B  \sum_{\mathbf{y}: (\sum_{k=1}^{i-1} y_k \geq 2) \land  (\sum_{k=1}^{i-1} y_k\cdot k = i) \land (y_j \geq 1)}  \left| \frac{(F'_j + z_j \Delta_j)^{y_j} - (F'_j)^{y_j}}{B^{y_j}} \right| \frac{1}{y_j!} \prod_{k \in [i-1]: k \neq j} \left(\frac{F'_k}{B}\right)^{y_k} \frac{1}{y_k!} \\
    \leq& B \sum_{\mathbf{y}: (\sum_{k=1}^{\tau} y_k \geq 2)\land (y_j \geq 1)} \left| \frac{(F'_j + z_j \Delta_j)^{y_j} - (F'_j)^{y_j}}{B^{y_j}} \right| \frac{1}{y_j!} \prod_{k \in [{\tau}]: k \neq j} \left(\frac{F'_k}{B}\right)^{y_k} \frac{1}{y_k!}.
\end{align*}
The inequality comes from the fact that the union of the disjoint sets of bucket assignments with a specific sum $i$ for $i \in \{j+1, \ldots, \tau\}$ is a subset of all bucket assignments.
Parameterizing by the number of elements $\ell$ with frequency $j$ and then separating out the term where $\ell = 1$,
\begin{align*}
    \sum_{i=j+1}^{\tau} f_j^i(\Delta)
    \leq& B \sum_{\ell=1}^\infty \sum_{\mathbf{y}: (\sum_{k=1}^{\tau} y_k \geq 2) \land ( y_j = \ell)} \left| \frac{(F'_j + z_j \Delta_j)^{\ell} - (F'_j)^{\ell}}{B^{\ell}} \right| \frac{1}{\ell!} \prod_{k \in [{\tau}]: k \neq j} \left(\frac{F'_k}{B}\right)^{y_k} \frac{1}{y_k!} \\
    =&B \sum_{\mathbf{y}: (\sum_{k=1}^{\tau} y_k \geq 2) \land ( y_j = 1)} \left| \frac{(F'_j + z_j \Delta_j) - F'_j}{B} \right| \prod_{k \in [{\tau}]: k \neq j} \left(\frac{F'_k}{B}\right)^{y_k} \frac{1}{y_k!} \\
    &+ B  \sum_{\ell=2}^\infty \sum_{\mathbf{y}:  y_j = \ell} \left| \frac{(F'_j + z_j \Delta_j)^{\ell} - (F'_j)^{\ell}}{B^{\ell}} \right| \frac{1}{\ell!} \prod_{k \in [{\tau}]: k \neq j} \left(\frac{F'_k}{B}\right)^{y_k} \frac{1}{y_k!} \\
    \leq&B \sum_{\mathbf{y}: (\sum_{k=1}^{\tau} y_k \geq 2) \land ( y_j = 1)} \left(\frac{\Delta_j}{F'_j}\right) \left(\frac{F'_j}{B}\right) \prod_{k \in [{\tau}]: k \neq j} \left(\frac{F'_k}{B}\right)^{y_k} \frac{1}{y_k!} \\
    &+ B  \sum_{\ell=2}^\infty \sum_{\mathbf{y}:  y_j = \ell} \left|\left(1 + \frac{z_j \Delta_j}{F'_j}\right)^\ell - 1\right|\left(\frac{F'_j}{B}\right)^{\ell} \frac{1}{\ell!} \prod_{k \in [\tau]: k \neq j} \left(\frac{F'_k}{B}\right)^{y_k} \frac{1}{y_k!}.
\end{align*}
Recall that due to our constraint on the sum over all $F'_k$, $\prod_{k \in [\tau]} e^{F'_k/B} \leq e^{C_1 |S|/B}$. We will introduce these exponential terms to make the terms within the summation resemble the Poisson density function:
\begin{align*}
    \sum_{i=j+1}^{\tau} f_j^i(\Delta)
    \leq& B  e^{C_1 |S|/B} \left(\frac{\Delta_j}{F'_j}\right)\sum_{\mathbf{y}: (\sum_{k=1}^{\tau} y_k \geq 2) \land (y_j = 1)}  \prod_{k \in [{\tau}]} \left(\frac{F'_k}{B}\right)^{y_k} \frac{e^{-F'_k/B}}{y_k!} \\
    &+ B   e^{C_1 |S|/B} \sum_{\ell=2}^\infty \left|\left(1 + \frac{z_j \Delta_j}{F'_j}\right)^\ell - 1\right| \sum_{\mathbf{y}: y_j = \ell} \prod_{k \in [{\tau}]} \left(\frac{F'_k}{B}\right)^{y_k} \frac{e^{-F'_k/B}}{y_k!}.
\end{align*}
The first summation above corresponds to the probability of a random variable distributed as $\poi(F'_j/B)$ equaling 1 and the sum over independent Poisson random variables distributed as $\poi(F'_k/B)$ for all other $k \neq j$ summing to at least one (so that the bucket contains at least two elements). Note that the sum over these independent Poisson random variables is itself a Poisson random variable with rate parameter at most $\frac{C_1|S| - F'_j}{B}$.
The second summation corresponds to the case where the $\poi(F'_j/B)$ variable equals $\ell$ and marginalizes over all assignments to the other variables:
\begin{align}
    \sum_{i=j+1}^{\tau} f_j^i(\Delta)
    \leq& B  e^{C_1 |S|/B} \left(\frac{\Delta_j}{F'_j}\right) \Pr\left(\poi\left(\frac{F'_j}{B}\right) = 1\right) \left(1 - \Pr\left(\poi\left(\frac{C_1 |S| - F'_j}{B}\right) = 0\right)\right) \nonumber \\
    &+ B   e^{C_1 |S|/B} \sum_{\ell=2}^\infty \left|\left(1 + \frac{z_j \Delta_j}{F'_j}\right)^\ell - 1\right| \Pr\left(\poi\left(\frac{F'_j}{B}\right) = \ell\right) \nonumber \\
    \leq& B  e^{C_1 |S|/B} \left(\frac{\Delta_j}{F'_j}\right) \left(\frac{F'_j}{B}\right) \left(1 - e^{-\frac{C_1 |S| - F'_j}{B}}\right) \nonumber \\
    &+ B  e^{C_1 |S|/B} \sum_{\ell=2}^\infty \left|\left(1 + \frac{z_j \Delta_j}{F'_j}\right)^\ell - 1\right| \left(\frac{F'_j}{B}\right)^\ell \frac{e^{-F'_j/B}}{\ell!}  \nonumber \\
    \leq& B e^{C_1 |S|/B} \left( \left( \frac{\Delta}{B}\right) \left(\frac{C_1 |S|}{B}\right)
    +  \sum_{\ell=2}^\infty \left|\left(1 + \frac{z_j \Delta_j}{F'_j}\right)^\ell - 1\right| \left(\frac{F'_j}{B}\right)^\ell \frac{1}{\ell!} \right). \label{e:sensitivitybound}
\end{align}

Proceeding by cases, we will show the following upper bound on the term which appears within the summation above: 
\begin{equation*}
    \left|\left(1 + \frac{z_j \Delta_j}{F'_j}\right)^\ell - 1\right| \left(\frac{F'_j}{B}\right)^\ell \frac{1}{\ell!} \leq \frac{2\Delta}{B} \left(\frac{C_1 |S|}{B}\right)^{\ell-1}.
\end{equation*}
Consider the case where $\Delta \geq F'_j$. If $z_j = -1$, then it must be the case that $\Delta_j = F'_j$ and the term evaluates to:
\begin{equation*}
    \left(1 - \left(1 - \frac{F'_j}{F'_j}\right)^\ell \right) \left(\frac{F'_j}{B}\right)^\ell \frac{1}{\ell!} = \left(\frac{F'_j}{B}\right)^\ell \frac{1}{\ell!}
    \leq \left(\frac{\Delta}{B}\right)^\ell \leq \frac{\Delta}{B} \left(\frac{C_1 |S|}{B}\right)^{\ell-1}
\end{equation*}
If $\Delta \geq F'_j$ and $z_j = 1$, using the fact that $\Delta_j = \Delta \leq C_1 |S|$,
\begin{equation*}
    \left(\left(1 + \frac{\Delta}{F'_j}\right)^\ell - 1\right) \left(\frac{F'_j}{B}\right)^\ell \frac{1}{\ell!}
    \leq \left(\frac{2\Delta}{F'_j}\right)^\ell \left(\frac{F'_j}{B}\right)^\ell \frac{1}{\ell!}
    = \left(\frac{\Delta}{B}\right)^\ell \frac{2^\ell}{\ell!}
    \leq \frac{2\Delta}{B} \left(\frac{C_1 |S|}{B}\right)^{\ell-1}.
\end{equation*}

Now, consider the case where $\Delta < F'_j$. Then, $\Delta_j = \Delta$.
If $z_j = -1$, via Bernoulli's inequality,
\begin{equation*}
    \left(1 - \left(1 - \frac{\Delta}{F'_j}\right)^\ell \right) \left(\frac{F'_j}{B}\right)^\ell \frac{1}{\ell!}
    \leq \left(\frac{\Delta \ell}{F'_j}\right) \left(\frac{F'_j}{B}\right)^\ell \frac{1}{\ell!}
    = \frac{\Delta}{B(\ell-1)!} \left(\frac{F'_j}{B}\right)^{\ell-1}
    \leq \frac{\Delta}{B} \left(\frac{C_1|S|}{B}\right)^{\ell-1}.
\end{equation*}
If $\Delta < F'_j$ and $z_j = 1$, then
\begin{equation*}
    \left(1 + \frac{\Delta}{F'_j}\right)^\ell - 1
    \leq 1 + \frac{\Delta}{F'_j} \left(\binom{l}{1} + \binom{l}{2} + \ldots + \binom{l}{l}\right) - 1
    = \frac{\Delta}{F'_j} \left(2^\ell - 1\right)
\end{equation*}
and
\begin{equation*}
    \left(\left(1 + \frac{\Delta}{F'_j}\right)^\ell - 1\right) \left(\frac{F'_j}{B}\right)^\ell \frac{1}{\ell!} 
    \leq \frac{\Delta}{F'_j} \left(2^{\ell} - 1\right)  \left(\frac{F'_j}{B}\right)^\ell  \frac{1}{\ell!} 
    \leq \frac{\Delta}{F'_j} \left(\frac{F'_j}{B}\right)^\ell 
    \leq \frac{2\Delta}{B} \left(\frac{C_1 |S|}{B}\right)^{\ell-1},
\end{equation*}
completing the case analysis.

Assuming that $C_1$ is chosen s.t.\ $C_1 |S| \leq 8B$, this results in the overall bound following from \cref{e:sensitivitybound}:
\begin{align*}
    \sum_{i=1}^{\tau} f_j^i(\Delta) &\leq B e^{C_1 |S|/B} \left( \frac{\Delta C_1 |S|}{B^2} + \left(\frac{2\Delta}{B}\right) \left(\sum_{\ell=2}^\infty \left(\frac{C_1 |S|}{B}\right)^{\ell-1}\right) \right)\\
    &= \Delta e^{C_1|S|/B} \left( \left(\frac{C_1 |S|}{B}\right) +  2 \sum_{\ell=1}^\infty \left(\frac{C_1 |S|}{B}\right)^\ell\right) \\
    &\leq \Delta e^{1/8} \left(\frac{1}{8} + \frac{2}{7}\right) \\
    &\leq  \frac{\Delta}{2}.
\end{align*}
This completes the proof of \cref{lem:sensitivity}.
\end{proof}

Recall that the sum of the $\zeta_i$ terms can be bounded in terms of $f_j^i$ as follows:
\begin{align*}
    \sum_{i=1}^{\tau} \zeta_i = \sum_{i=1}^{\tau} \sum_{k=1}^i \sum_{1 \leq j_1 < j_2 < \ldots < j_k \leq i-1} f_{j_k}^i \left( f_{j_{k-1}}^{j_k} \left( \ldots \left( f_{j_2}^{j_3} \left( f_{j_1}^{j_2} \left(\gamma_{j_1}\right) \right)\right) \right)\right).
\end{align*}
We will show that as $k$ increases, the impact of the $\gamma_j$ terms diminish geometrically due to the sensitivity of $\hat{r}_i(\cdot)$ as described by \cref{lem:sensitivity}.
Let $\alpha_{j,k}$ be the sum over all terms from above for a fixed $k$ and for $j_k = j$, i.e., the effect of random error in $F'_j$ on all estimates of $\hat{r}_i$ for $i > j$ via dependency chains of length $k$:
\begin{equation}
    \alpha_{j,k} = \sum_{i=j+1}^{\tau} \sum_{1 \leq j_1 < j_2 < \ldots < j_k = j} f_{j_k}^i \left( f_{j_{k-1}}^{j_k} \left( \ldots \left( f_{j_2}^{j_3} \left( f_{j_1}^{j_2} \left(\gamma_{j_1}\right) \right)\right) \right)\right).
\end{equation}
Note that 
\begin{equation} 
\label{eq:zeta-alpha}
\sum_{j=1}^{\tau} \sum_{k=1}^{\tau} \alpha_{j,k} = \sum_{i=1}^{\tau} \zeta_i.
\end{equation}
We will inductively prove the following inequality about the $\alpha_{j,k}$ terms.
\begin{lemma}\label{lem:alphajk}
\begin{equation*}
    \sum_{j=1}^{\tau} \alpha_{j, k} \leq \frac{1}{2^k} \sum_{i=1}^{\tau}  \gamma_i.
\end{equation*}
\end{lemma}

\begin{proof}
For the base case, consider $k=1$. Directly from \cref{lem:sensitivity}, we know that
\begin{equation*}
    \alpha_{j, 1} = \sum_{i=j+1}^{\tau} f_j^i(\gamma_j) \leq \frac{\gamma_j}{2}.
\end{equation*}
Therefore,
\begin{equation*}
    \sum_{j=1}^{\tau} \alpha_{j, 1} \leq \sum_{j=1}^{\tau} \frac{\gamma_j}{2},
\end{equation*}
as required.
Now, consider some $k > 1$ and assume the inequality holds for $k-1$.
Consider any fixed $j$:
\begin{equation*}
    \alpha_{j, k} = \sum_{i=j+1}^{\tau} \sum_{1 \leq j_1 < j_2 < \ldots < j_k = j} f_{j_k}^i \left( f_{j_{k-1}}^{j_k} \left( \ldots \left( f_{j_2}^{j_3} \left( f_{j_1}^{j_2} \left(\gamma_{j_1}\right) \right)\right) \right)\right)
\end{equation*}
Let $T_{j,k} = \{ (j_1, \ldots , j_k): 1 \leq j_1  \leq j_2 \leq \ldots \leq j_k=j\}$ be the set of all valid sequences of length $k$ ending in $j$ and for $t \in T_{j,k}$, let
\begin{equation*}
    \beta_{t} = f_{j_{k-1}}^{j_k}\left(f_{j_{k-2}}^{j_{k-1}} \left( \ldots \left(f_{j_1}^{j_2}(\gamma_1)\right)\right)\right)
\end{equation*}
be the corresponding composition of $f$'s.
Then,
\begin{equation*}
    \alpha_{j, k} = \sum_{i=j+1}^{\tau} \sum_{t \in T_{j,k}} f_{j}^i\left(\beta_{t}\right) = \sum_{t \in T_{j,k}} \sum_{i=j+1}^{\tau}  f_{j}^i\left(\beta_{t}\right).
\end{equation*}
By \cref{lem:sensitivity}, for any given $t \in T_{j,k}$,
\begin{equation*}
    \sum_{i=j+1}^{\tau} f_j^i(\beta_t) \leq \frac{\beta_t}{2}
\end{equation*}
which in turn implies a bound on $\alpha_{j, k}$:
\begin{equation*}
    \alpha_{j, k} \leq \sum_{t \in T_{j,k}} \frac{\beta_t}{2}.
\end{equation*}
Rewriting $\sum_{j=1}^{\tau} \alpha_{j, k-1}$ in terms of the $\beta_t$ variables:
\begin{align*}
    \sum_{j=1}^{\tau} \alpha_{j, k-1} &=\sum_{j=1}^{\tau}\sum_{i=j+1}^{\tau} \sum_{1 \leq j_1 < j_2 < \ldots < j_{k-1} = j} f_{j_{k-1}}^i \left( f_{j_{k-2}}^{j_{k-1}} \left( \ldots \left( f_{j_2}^{j_3} \left( f_{j_1}^{j_2} \left(\gamma_{j_1}\right) \right)\right) \right)\right) \\
    &= \sum_{i=1}^{\tau} \sum_{1 \leq j_1 < j_2 < \ldots < j_{k-1} < j_k = i} f_{j_{k-1}}^{j_k}\left( f_{j_{k-2}}^{j_{k-1}} \left( \ldots \left( f_{j_2}^{j_3} \left( f_{j_1}^{j_2} \left(\gamma_{j_1}\right) \right)\right) \right)\right)\\
    &= \sum_{i=1}^{\tau} \sum_{t \in T_{i,k}}\beta_t.
\end{align*}
Applying the inductive hypothesis, 
\begin{equation*}
    \sum_{j=1}^{\tau} \alpha_{j, k-1} = \sum_{i=1}^{\tau} \sum_{t \in T_{i,k}}\beta_t \leq \frac{1}{2^{k-1}} \sum_{i=1}^{\tau} \gamma_i.
\end{equation*}
Combining the inequalities,
\begin{equation*}
    \sum_{j=1}^{\tau} \alpha_{j,k} \leq \sum_{j=1}^{\tau} \sum_{t \in T_{j,k}}  \frac{\beta_t}{2}
    \leq \frac{1}{2} \left(\frac{1}{2^{k-1}} \sum_{i=1}^{\tau} \gamma_i\right)
    = \frac{1}{2^k} \sum_{i=1}^{\tau} \gamma_i,
\end{equation*}
as required.
\end{proof}

Now, we can complete the proof of \cref{lem:errprop} by finally bounding the total sum of the $\zeta_i$ terms by the $\gamma_i$ terms.
\begin{proof}[Proof of \cref{lem:errprop}]
From \cref{eq:zeta-alpha} and \cref{lem:alphajk}:
\begin{equation*}
    \sum_{i=1}^{\tau} \zeta_i = \sum_{k=1}^{\tau} \sum_{j=1}^{\tau} \alpha_{j, k}
    \leq \sum_{k=1}^{\tau} \frac{1}{2^k} \sum_{i=1}^{\tau} \gamma_i
    \leq \sum_{i=1}^{\tau} \gamma_i.
\end{equation*}
\end{proof}

In summary, as the errors propagate (via composition of the $f_j^i$'s), their magnitudes diminish geometrically, allowing us to bound the total propagated error by the initial errors.

Now, we will bound the $\gamma_i$ terms which, via \cref{lem:errprop}, bound $\sum_{i=1}^{\tau} |\hat{F}_i - F_i|$ up to constant factors.
\begin{lemma}\label{lem:varerror} 
If $errortype=m$ and $B = \Theta(\log \tau / \epsilon^2)$,
\begin{equation*}
    \sum_{i=1}^{\tau} \gamma_i = O\left(\frac{\log \tau}{\eps}\right),
\end{equation*}
and if $errortype=D$, $\tau=O(1)$, and $B = \Theta(1 / \epsilon^2)$,
\begin{equation*}
    \sum_{i=1}^{\tau} \gamma_i =O\left(\frac{1}{\eps}\right),
\end{equation*}
both holding with constant probability.
\end{lemma}

\begin{proof}
Recall the expressions for our estimator
\begin{equation*}
    \hat{F}_i = \max\{b_i e^{\hat{S}/B} - \hat{r}_i(\hat{F}_1, \ldots, \hat{F}_{i-1}), 0\},
\end{equation*}
and for the ground truth
\begin{equation*}
    F_i = \E[s_i] e^{|S|/B} = \E[b_i] e^{|S|/B} - \E[r_i] e^{|S|/B} = \E[b_i] e^{|S|/B} - \hat{r}_i(F_1, \ldots, F_{i-1}).
\end{equation*}
The contribution to the error measured by $\gamma_i$ is
\begin{equation}\label{e:bdeviation}
    \gamma_i = |\E[b_i] e^{|S|/B} - b_i e^{\hat{S}/B} | \leq \E[b_i]\left|e^{|S|/B} - e^{\hat{S}/B}\right| + e^{\hat{S}/B} \left|b_i - \E[b_i]\right|.
\end{equation}
Let $\delta_S = \hat{S} - |S|$. Then, the first term can be expressed as
\begin{align*}
    \E[b_i]\left|e^{|S|/B} - e^{\hat{S}/B}\right| &= \E[b_i]e^{|S|/B}\left|1 - \frac{e^{\hat{S}/B}}{e^{{|S|/B}}}\right| \\
    &=\E[b_i]e^{|S|/B}\left|1 - e^{\delta_S/B}\right| \\
    &\leq \E[b_i]e^{|S|/B}\left|\frac{\delta_S}{B}\right|.
\end{align*}
Recall from \cref{lem:hatS}, with constant probability, $|\delta_S| =O(\sqrt{|S|})$. Under this assumption, and since $|S|=O(B)$ (\cref{lem:S}), the sum over all of such terms can be bounded by
\begin{equation*}
    \sum_{i=1}^{\tau} \E[b_i]\left|e^{|S|/B} - e^{\hat{S}/B}\right| = O\left(\sum_{i=1}^{\tau} \frac{\sqrt{|S|}E[b_i]}{B}\right) = O\left(\frac{|S|^{3/2}}{B}\right) = O\left(\sqrt{B}\right). 
\end{equation*}

Now, let us consider the other terms in \cref{e:bdeviation} which account for the deviation of $b_i$ from its expectation:
\begin{equation*}
    \sum_{i=1}^{\tau} e^{\hat{S}/B} |b_i - \E[b_i]| = O\left(\sum_{i=1}^{\tau}  |b_i - \E[b_i]|\right) = O\left(\sum_{i=1}^{\tau}  \sqrt{\Var(b_i)}\right)
\end{equation*}
where the last inequality holds with constant probability.
Let $Z_{i, k}$ be an indicator random variable for the event that bucket $k$ has summed count exactly $i$.
Then, $b_i = \sum_{k=1}^B Z_{i, k}$. Further, the $Z_{i, k}$'s are mutually independent due to Poissonization, so $b_i$ has a Binomial distribution.
Therefore, $\Var(b_i) \leq \E[b_i]$.

As the sum of the counts over all buckets is the total mass of the sampled elements,
\begin{equation*}
    \sum_{i=1}^\tau \E[b_i] i = \frac{|S| m}{D}.
\end{equation*}
In the case where $errortype=m$, this evaluates to $O(B)$.
In either case, $\sum_i b_i \leq S = O(B)$.
By application of \cref{lem:sumofsqrt}, if $errotype=m$,
\begin{equation*}
    \sum_{i=1}^{\tau} \sqrt{\E[b_i]} = O(\sqrt{B \log \tau}),
\end{equation*}
and if $errortype=D$,
\begin{equation*}
    \sum_{i=1}^{\tau} \sqrt{\E[b_i]} = O(\sqrt{B }).
\end{equation*}

Assuming $errortype=m$ and $B = \Theta(\log \tau / \epsilon^2)$, the total contribution of errors from both terms is
\begin{equation*}
    \sum_{i=1}^{\tau} \gamma_i = O\left(\sqrt{B} + \sqrt{B\log \tau}\right) = O\left(\frac{\log \tau}{\eps}\right),
\end{equation*}
with constant probability.

Assuming $errortype=D$ and $B = \Theta(1 / \epsilon^2)$, the total contribution of errors from both terms is
\begin{equation*}
    \sum_{i=1}^{\tau} \gamma_i = O\left(\sqrt{B}\right) = O\left(\frac{1}{\eps}\right),
\end{equation*}
with constant probability.
\end{proof}

This completes the analysis of inversion procedure, showing that \cref{alg:dp} returns a set of frequencies $\hat{F}_1, \ldots \hat{F}_{\tau}$ close in $L_1$ distance to the empirical profile of our samples $F_1, \ldots, F_{\tau}$.

\subsection{Pseudorandomness}\label{subsec:pseudorandom}
So far, in our analysis, we have assumed that our hash functions are fully random. In order to bound the space required to store the hash functions, we will now argue that the guarantees of the algorithm only require limited independence.
We will require a lemma from prior work on applying Nisan's PRG to streaming algorithms.
\begin{lemma}[Lemma 3 from \cite{indyk2006stable}]\label{lem:nisan}
Consider an algorithm $\mathcal{A}$ that, given a stream $\mathcal{S}$ of elements $x$,
and a function $f : [n] \times \{0, 1\}^R  \rightarrow [\text{poly}(M)] \times [\text{poly}(M)]$, does the following:
\begin{itemize}
    \item Set $\mathcal{O}= (0, 0)$; Initialize length-$R$ chunks $R_0, \ldots ,R_n$ of independent random bits
    \item For each new element $x$, perform $\mathcal{O}= \mathcal{O}+ f (x, R_x)$
    \item Output $A(S) = \mathcal{O}$
\end{itemize}
Assume that the function $f ( \cdot , \cdot )$ is computed by an algorithm using $O(C + R)$ space and $O(T)$ time. Then there is an algorithm $\mathcal{A}'$ producing output $\mathcal{A}'(\mathcal{S})$, that uses only $O(C + R + log(Mn))$ bits of storage and $O([C + R + log(Mn)] log(n R))$
random bits, such that
\begin{equation*}
    \Pr[\mathcal{A}(\mathcal{S}) \neq \mathcal{A'}(\mathcal{S})] \leq 1/\text{poly}(n)
\end{equation*}
over some joint probability space of randomness of $\mathcal{A}$ and $\mathcal{A}'$. Then, the algorithm $\mathcal{A}'$ uses $O(T + log(n R))$ arithmetic operations per each element $x$.
\end{lemma}
Note several minor changes from the lemma as it appears in \cite{indyk2006stable}: the original lemma includes a parameter for the size of stream updates which we omit as all of our updates are increments, the original lemma considers a one dimensional output while we consider a two dimensional output, and the original lemma bounds the bias by $1/n$ while the bias can in fact by bounded by $1/\text{poly}(n)$ without changing the asymptotic result.

\begin{lemma}\label{lem:hashfn}
$O\left(\log^3(B) + \log n\right)$ bits of randomness suffice for the hash functions used in \cref{alg:profile}.
\end{lemma}
\begin{proof} 
The randomness of \cref{alg:profile} is realized through the following hash functions $g_1$, $g_2$, $z$, $h_1, \ldots, h_H$.
The hash function $g_1: [n] \rightarrow [n]$ is used to sample elements based on the least significant bit of their hashed values. The randomness of this hash function is used in analysis of the space needed to store $A$ as well as $|S|$ and $F_1, \ldots, F_m$. The analysis of those quantities only use first and second moments, so pairwise independence suffices for $g_1$.

The hash function $g_2: [n] \rightarrow [T]$ maps the sampled elements down to a domain of size $\Theta(B^2)$. We only require that with constant probability, there are no collisions between sampled elements. As with constant probability $|S| = O(\sqrt{T})$ (see \cref{lem:S}), this holds under pairwise independence.

The hash function $z: [T] \rightarrow \N \cup \{0\}$ maps sampled stream elements to the outcome of a $\poi(1)$ random variable.\footnote{For example via inverse transform sampling \url{https://en.wikipedia.org/wiki/Inverse_transform_sampling}.}
As long as $z$ is pairwise independent, the total number and mass of (copied) samples will be concentrated about their expectation.

The hash functions $h_1, \ldots, h_H: [T] \rightarrow [B]$ map sampled stream elements to a bucket in the array $A$.
These hash functions (as well as $z$) affect $b_i$, the number of buckets with count exactly $i$ (as well as their sum, $G$, the number of nonempty buckets). 

In our analysis, we use the first and second moments of the $b_i$'s to bound the error of our algorithm.
Let $X_k$ be a random variable for the count in bucket $k$.
The first two moments of $b_i$ depend only on the joint distribution of $(X_k, X_{k'})$. We will use \cref{lem:nisan} to show that limited randomness suffices to simulate this distribution up to small bias.
Consider the following setting of the parameters in \cref{lem:nisan} where the algorithm $\mathcal{A}$ computes the pair $(X_k, X_{k'})$ given fully random hash functions:
\begin{itemize}
    \item $n = O(B^2 \log^2 \tau)$ as we are conditioning on the outcomes of $g_1, g_2$ which have generated $S$ distinct items for our stream coming from a domain of size $O(B^2 \log^2 \tau)$.
    \item $R = O(\log^2 B)$ is the number of random bits required to generate fully random $z(x)$, $h_1(x)$, $\ldots$, $h_H(x)$ as the range of these hash functions is $B$ and there are $H=O(\log B)$ such hash functions.
    \item $M = \tau$ as the counts of a bucket $X_k$ is at most $\tau$.
    \item $f(x, R_x)$ is computed by $\mathcal{A}$ as follows. Set the algorithm's output to be $\mathcal{O} = (a, b) = (0,0)$. $\mathcal{A}$ uses the initial bits of $R_i$ to generate $z(x)$. Then, $\mathcal{A}$ uses each of the next chunks of $O(\log(B))$ bits in $R_i$ to generate $h_1(x), \ldots, h_H(x)$. On generating $h_i(x)$, $\mathcal{A}$ first checks if $z(x) < i$ and if so, stops early and returns $(a,b)$. Otherwise, $\mathcal{A}$ checks if $h_i(x) = k$ and if so increments $a$ by one (unless the count exceeds $\tau$). Finally, $\mathcal{A}$ checks if $h_i(x) = k'$ and if so increments $b$ by one (unless the count exceeds $\tau$). Via this process, $\mathcal{A}$ will recover the true counts of $(X_k, X_{k'})$.
    \item $C$, the space used by $\mathcal{A}$, is $O(\log B)$. In processing a given stream element, $\mathcal{A}$ stores the outcome of $z(x)$, the current $h_i(x)$, $\mathcal{O}$, and the current increment to $\mathcal{O}$. In total, these are a constant number of quantities, each taking up at most $O(\log B)$ space.
\end{itemize}
It follows from application of \cref{lem:nisan} that the algorithm $\mathcal{A'}$ using Nisan's PRG requires only
\begin{equation*}
    O(C + R + \log(Mn))\log(nR) = O\left(\log B + \log^2(B) + \log B) \log B\right) = O(\log^3 B)
\end{equation*}
random bits to approximate the joint distribution of $(X_k, X_{k'})$ up to bias $\text{poly}(\eps)$:
\begin{equation*}
    \Pr[\mathcal{A}(\mathcal{S}) \neq \mathcal{A}'(\mathcal{S})] \leq 1 / \text{poly}(B).
\end{equation*}
For large enough $\text{poly}(B)$, this implies that the expectations and variances of the $b_i$'s using limited randomness will be correct up to a small factor of $\eps$, which is enough for our analysis.

As the hash functions $g_1$ and $g_2$ only require constant independence, they can be stored and queried in $O(\log n)$ bits~\cite{carter1979hash}.
So, the total space required for storing randomness is $O(\log^3(B) + \log n)$.
\end{proof}

\subsection{Putting it all together}
We are now ready to prove the main theorems.

\subsubsection{Error type \texorpdfstring{$m$}{m}}
\begin{proof}[Proof of \cref{thm:algm}]
\item
\paragraph*{Correctness}
Recall that our estimator is
\begin{equation*}
    \hat{\phi}_i = \frac{\hat{D}}{\hat{S}} \hat{F}_i.
\end{equation*}
From \cref{lem:errprop,lem:varerror}, with constant probability,
\begin{equation*}
    \sum_{i=1}^{\tau} |F_i - \hat{F}_i| = O\left(\frac{\log \tau}{\eps}\right).
\end{equation*}
Also, recall that by \cref{lem:empiricalprofile2},  with constant probability,
\begin{equation}\label{e:dm-error}
    \sum_{i=1}^{\tau} \left|\phi_i - \frac{D}{|S|} F_i\right| = O(\eps m).
\end{equation}
With constant probability, $|S| = \Omega\left(\frac{D \log\tau}{m \eps^2}\right)$ by \cref{lem:S}, $|\hat{S} - |S|| = O(\sqrt{|S|})$ by \cref{lem:hatS}, and $|\hat{D} - D| \leq \eps D/10$ by the correctness of the distinct elements sketch (e.g., \cite{kane2010optimal}). In addition, $\log \tau / \eps = O(\eps B) = o(\eps m)$.
Under these conditions,
\begin{align}
    \sum_{i=1}^{\tau} \left|\frac{D}{|S|} F_i - \frac{\hat{D}}{\hat{S}} \hat{F}_i\right| 
    &\leq \sum_{i=1}^{\tau} \frac{D}{|S|}\left|F_i - \hat{F}_i\right| + \left|\frac{D}{|S|} - \frac{\hat{D}}{\hat{S}}\right| \hat{F}_i \nonumber \\
    &\leq \frac{D}{|S|} \cdot O\left(\frac{\log\tau}{\eps}\right) + \sum_{i=1}^{\tau}  \frac{\eps D}{5\hat{S}} \hat{F}_i \nonumber \\
    &= O\left(\frac{m \eps^2}{\log\tau} \left(\frac{\log\tau}{\eps}\right)\right) + \frac{\eps D}{5} \nonumber \\
    &= O(\eps m).\label{e:finalerr}
\end{align}

Recall that $\tau = O(1/\eps)$. Implicitly, we will predict $\hat\phi_i = 0$ for all $i > \tau$. As there are at most $m/\tau$ elements with frequency greater than or equal to $\tau$, this contributes error at most $O(\eps m)$.
So, with appropriately chosen constants and with constant probability, triangle inequality with \cref{e:dm-error} and \cref{e:finalerr} gives
\begin{equation*}
    \sum_{i=1}^m \left|\phi_i - \hat{\phi}_i\right| \leq \eps m,
\end{equation*}
as required.

\paragraph*{Space}
Now, we will analyze the space used by the algorithm while processing stream updates.
By \cref{lem:counts}, the array $A$ can be maintained in $O(B) = O(\log (1/\eps) / \eps^2)$ bits of space.
By \cref{lem:hashfn}, it suffices to store $O(\log^3(B) + \log n) = O(\log^3(1/\eps) + \log n)$ bits of randomness.
As each stream update occurs, we must store its identity in $O(\log n)$ bits of space. Storing the length of the stream na\"ively takes $O(\log m)$ bits. Given a $\poly (m)$ upper bound on the stream length, this can be reduced to $O(\log\log m)$ bits using the Morris+ algorithm of~\cite{nelson2022optimalmorris} to maintain a small constant approximation to the stream length (which is all we require) at all times with small constant failure probability.
The distinct elements sketch can be maintained in  $O(1/\eps^2 + \log n)$ bits of space \cite{kane2010optimal}.
Therefore, the total space usage of our algorithm is, with constant probability, $O(\log (1/\eps)/\eps^2 + \log n + \log \log m)$ bits.

\paragraph*{Update time}
Each time a stream element appears, hashing it for sampling and to compress its ID each require a constant number of operations. Without accounting for the pseudorandom generator, Poissonization and incrementing bucket counts takes $O(1)$ expected amortized time as each element is copied $O(1)$ times in expectation and array updates can be done in expected amortized constant time~\cite{blandford2008CompactDF} (by also storing the (level, counter) dictionaries in cells of the array $A$ as variable-bit-length dictionaries). Deletions of (level, counter) pairs take constant amortized time: they take $O(B)$ time each time the size of the stream doubles relative to $B$. Then, by \cref{lem:nisan} and the parameters used in \cref{lem:hashfn}, these updates take $O(1 + \log\left(B^2 \log^2 B \log^2 \tau\right)) = O(\log(1/\eps))$ time accounting for Nisan's PRG. Finally, updating the distinct elements sketch takes $O(1)$ time~\cite{kane2010optimal}. So, the total update time is $O(\log(1/\eps))$.

\paragraph*{Post-processing}
The dynamic program for post-processing maintains $O(1/\eps^2)$ cells, and the computation takes time polynomial in $1/\eps$. Furthermore, all entries in the DP table are positive. It follows that performing the computation with $O(\log(1/\eps))$ bits of precision per cell suffices to calculate the answer with a multiplicative error of $1+\eps^{O(1)}$. So, the total space usage is $O(\log (1/\eps)/\eps^2)$.

For a given cell in the dynamic program in \cref{alg:dp}, calculating $\DP[i, x]$ requires $O(i^2/x)$ operations due to the nested sum. As we are only concerned with cells where $x \leq i$, filling the row in the table corresponding to a fixed $i$ takes $\sum_x O(i^2/x) = O(i^2 \log i)$ time. Summing over all $i \in \{1, \ldots, \tau\}$, the total post-processing time is bounded by $O(\log (1/\eps)/\eps^3)$.
\end{proof}

\subsubsection{Error type \texorpdfstring{$D$}{D}} 

\begin{proof}[Proof of \cref{thm:algD}]
\item
\paragraph*{Correctness}
Recall that our estimator is
\begin{equation*}
    \hat{\phi}_i = \frac{\hat{D}}{\hat{S}} \hat{F}_i
\end{equation*}
and that $\tau = O(1)$.
From \cref{lem:errprop,lem:varerror}, with constant probability,
\begin{equation*}
    \sum_{i=1}^{\tau} |F_i - \hat{F}_i| = O\left(\frac{1}{\eps}\right).
\end{equation*}
Also, recall that by \cref{lem:empiricalprofile2},  with constant probability,
\begin{equation}\label{e:dm-errorD}
    \sum_{i=1}^{\tau} \left|\phi_i - \frac{D}{|S|} F_i\right| = O(\eps D).
\end{equation}
With constant probability, $D/|S| = \Theta\left(D\eps^2\right)$ by \cref{lem:S}, $|\hat{S} - |S|| = O(\sqrt{|S|})$ by \cref{lem:hatS}, and $|\hat{D} - D| \leq \eps D/10$ by the correctness of the distinct elements sketch (e.g., \cite{kane2010optimal}). 
Under these conditions,
\begin{align}
    \sum_{i=1}^{\tau} \left|\frac{D}{|S|} F_i - \frac{\hat{D}}{\hat{S}} \hat{F}_i\right| 
    &\leq \sum_{i=1}^{\tau} \frac{D}{|S|}\left|F_i - \hat{F}_i\right| + \left|\frac{D}{|S|} - \frac{\hat{D}}{\hat{S}}\right| \hat{F}_i \nonumber \\
    &\leq \frac{D}{|S|} \cdot O\left(\frac{1}{\eps}\right) + \sum_{i=1}^{\tau}  \frac{\eps D}{5\hat{S}} \hat{F}_i \nonumber \\
    &= O\left( \frac{D\eps^2}{\eps}\right) + \frac{\eps D}{5} \nonumber \\
    &= O(\eps D).\label{e:finalerrD}
\end{align}
So, with appropriately chosen constants and with constant probability, triangle inequality with \cref{e:dm-errorD} and \cref{e:finalerrD} gives
\begin{equation*}
    \sum_{i=1}^\tau \left|\phi_i - \hat{\phi}_i\right| \leq \eps D,
\end{equation*}
as required.

\paragraph*{Space}
Now, we will analyze the space used by the algorithm while processing stream updates.
By \cref{lem:counts}, the array $A$ can be maintained in $O(B) = O(1 / \eps^2)$ bits of space.
By \cref{lem:hashfn}, it suffices to store $O(\log^3(B) + \log n) = O(\log^3(1/\eps) + \log n)$ bits of randomness.
As each stream update occurs, we must store its identity in $O(\log n)$ bits of space. We can avoid storing the length of the stream $t$ for $errortype=D$ as we never use the stream length.
The distinct elements sketches can be maintained in $O(1/\eps^2 + \log n)$ and $O(\log n)$  bits of space \cite{kane2010optimal, blasiok2020distinctelements}.
Therefore, the total space usage of our algorithm is, with constant probability, $O( 1/\eps^2 + \log n)$ bits.

\paragraph*{Update time}
From the proof above of \cref{thm:algm}, the expected amortized update time is at least $O(\log(1/\eps))$. In addition, we have to update the tracking distinct elements sketch, which can be done in $O(\log \log n)$ time~\cite{blasiok2020distinctelements}.

\paragraph*{Post-processing}
As in the proof above of \cref{thm:algm}, but with $\tau=O(1)$, the post-processing space usage is bounded by $O( \log(1/\eps))$ and the post-processing time is bounded by $O(1)$ (assuming arithmetic with $O(\log(1/\eps))$ bits of precision can be done in constant time) as the DP table has constant size.
\end{proof}

\section{Lower Bounds}
\label{s:lower}
\subsection{Preliminaries}
For a function $f: \mathcal{X} \times \mathcal{Y} \rightarrow \{0,1\}$, we use $D_{\mu, \delta}^{1-way}(f)$ to denote the minimum communication of a protocol, in which a single message from Alice to Bob is sent, for solving $f$ with probability at least $1-\delta$, where the probability is taken over both the coin tosses of the protocol and the input distribution $\mu$. We consider the indexing function $f = ${\sc IND} in which $\mathcal{X} = \{0,1\}^n$, $\mathcal{Y} = [n] = \{1, 2, \ldots, n\}$, and $f(x,y) = x_y$. The following theorem is well-known:

\begin{theorem}\cite{DBLP:conf/stoc/KremerNR95}\label{thm:IND}
If $\mu$ is the product uniform distribution on $\{0,1\}^n \times [n],$ then 
$D_{\mu, 1/3}^{1-way}({\sc IND}) = \Omega(n)$. 
\end{theorem}

We need the following public-coin reduction from {\sc IND} to the so-called {\sc Gap-Hamming} Problem. In the following, $\Delta(w,z)$ denotes the Hamming distance between binary strings $w$ and $z$, that is, the number of positions in which the strings differ.  

\begin{theorem}\cite{DBLP:journals/toc/JayramKS08} (also Lemma 62 of \cite{woodruff2007efficient} with the constant in the success probability appropriately adjusted) \label{thm:INDtoGHD}
Let $x \in \{0,1\}^t$ and $y \in [t]$ be arbitrary. Let $R$ be a public coin and partition it into $t$ chunks $R^1, \ldots, R^t$, each of size $s = C t$ for a certain constant $C > 0$. Let $z = R^y$ and let $w = \textrm{Majority}(R^{\ell} \mid x_{\ell} = 1)$. Then there are constants $\alpha > \beta$ such that with probability at least $99/100$:
\begin{enumerate}
\item If $x_{y} = 0$, then $\Delta(w, z) > s/2 - \beta \sqrt{s}$, and
\item If $x_{y} = 1$, then $\Delta(w, z) < s/2 - \alpha \sqrt{s}$
\end{enumerate}
Here $w = \textrm{Majority}(R^{\ell} \mid x_{\ell} = 1)$ is an $s$-bit string, and $w_j$ is equal to the majority of the bits $R^{\ell}_j$ for which $x_{\ell} = 1$, where if exactly half the bits are $1$, we say the majority is $1$ for concreteness. 
\end{theorem}

\subsection{Reduction for error type \texorpdfstring{$m$}{m}}
\begin{theorem}\label{thm:lb_m}
    Any single pass streaming algorithm which outputs an additive $\epsilon m$ $L_1$ approximation to the profile with success probability at least $99/100$ requires $\Omega(\log(1/\eps)/\eps^2)$ bits of memory. 
\end{theorem}

\begin{proof}
We reduce the $L_1$ profile problem with the guarantee in \cref{e:l1_m} from the {\sc IND} problem with $n = (\log (1/\eps))/\epsilon^2$. We will assume the streaming algorithm succeeds with probability at least $99/100$ over its internal randomness. Our reduction is randomized, and uses \cref{thm:INDtoGHD} at different scales, i.e., values $t$ for which it is invoked. 

Alice, given an $x \in \{0,1\}^n,$ partitions $x$ into $r = \log (1/\eps)$ bit strings $x^1, x^2, \ldots, x^r$, each of length $n/r = 1/\epsilon^2$. For each $j \in [r],$ Alice
further partitions $x^j$ into strings $x^{j, \ell, k}$, where $\ell \in [2^{j-1}+1, 2^j]$ is an integer multiple of $4$ in this range, and
each $k \in [2^j]$. The index $j$ represents a ``geometric scale" of possible frequencies, the index $\ell$ represents one particular frequency in this geometric scale, and the index $k$ represents one particular domain item which could have frequency $\ell$ in the stream, as described below. 

Each string $x^{j, \ell, k}$ will be of length $t_j = 1/(2^{2j} \epsilon^2)$, which we can assume is an integer by adjusting $\epsilon$ by a constant factor. For each string $x^{j, \ell, k}$, Alice independently applies the public-coin reduction of Theorem \ref{thm:INDtoGHD} to create a string $w^{j, \ell, k} \in \{0,1\}^{C t_j}$, where $C$ is as in Theorem \ref{thm:INDtoGHD}. 

Bob, given a $y \in [n],$ interprets it as a triple $(j^*, \ell^*, k^*)$. Bob will use $k^*$ in what follows, but will temporarily ignore $j^*$ and $\ell^*$. For each possible $(j, \ell, k^*)$, Bob uses the same part of the public coin that Alice used to generate her string $w^{j, \ell, k}$ in Theorem \ref{thm:INDtoGHD}, for generating his own string $z^{j, \ell, k^*} \in \{0,1\}^{C t_{j}}$. 

For each $j \in [r]$, each $\ell \in [2^{j-1}+1, 2^j]$ which is an integer multiple of $4$, and each $k \in [2^j]$, we associate a chunk of
of $Ct_j = \frac{C}{\eps^2 2^{2j}}$ domain elements $Q_{\ell, k}^j[1], \ldots, Q_{\ell, k}^j[\frac{C}{\eps^2 2^{2j}}]$. The total size of the domain will thus be $n = \sum_{j=1}^{\log (1/\eps)} \Theta(2^{2j}) \left(\frac{C}{\eps^2 2^{2j}}\right) = \Theta \left (\frac{\log (1/\eps)}{ \eps^2} \right ).$
The stream is formed as follows:

\begin{enumerate}
    \item For each $(j, \ell, k)$, {\bf Alice}  associates the $C t_j$ domain elements $Q_{\ell, k}^j$ with the coordinates of her string $w^{j, \ell, k} \in \{0,1\}^{C t_j}$. Namely, Alice adds one copy of each domain element that corresponds to a position in $w^{j, \ell, k}$ that is $1$.

    \item For each $(j, \ell, k^*)$, {\bf Bob} adds $\ell$ copies of each domain element corresponding to a position in $z^{j, \ell, k^*}$ that is $0$, and adds $\ell+2$ copies of each domain element corresponding to a position in $z^{j, \ell, k^*}$ that is $1$. 
    
\end{enumerate}

Let $\phi$ be the true profile of this stream. For $k \neq k^*$, the domain elements $Q_{\ell, k}^j$ have frequency $0$ or $1$. The domain elements $Q_{\ell, k^*}^{j}$ have frequency in $\{\ell, \ell+1, \ell+2, \ell+3\}$. In particular, elements in $Q_{\ell, k^*}^{j}$ have frequency $\ell+1$ if and only if their corresponding index in $z^{j, \ell, k^*}$ was $0$ but in $w^{j, \ell, k^*}$ was $1$,  and have frequency $\ell+2$ if and only if the corresponding index in $z^{j, \ell, k^*}$ was $1$ and in $w^{j, \ell, k^*}$ was $0$.
Note also that $\ell$ skips in integer multiples of $4$, so the sets $\{\ell, \ell+1, \ell+2, \ell+3\}$ and $\{\ell', \ell'+1, \ell+2, \ell'+3\}$ are disjoint for distinct $\ell$ and $\ell'$ considered by Alice and Bob. Therefore, the Hamming distance between $w^{j, \ell, k^*}$ and $z^{j, \ell, k^*}$ is exactly $\phi_{\ell+1} + \phi_{\ell+2}$.

Since with probability at least $99/100$ over the randomness in the public coin used for the randomized reduction for $j^*$, $\ell^*$ and $k^*$, we have $\Delta(w^{j^*, \ell^*, k^*}, z^{j^*, \ell^*, k^*}) > s/2 - \beta \sqrt{s}$ if $x_y = 0$, and have  $\Delta(w^{j^*, \ell^*, k^*}, z^{j^*, \ell^*, k^*}) < s/2 - \alpha \sqrt{s}$ if $x_y = 1$, where $\alpha > \beta$ are the constants in Theorem \ref{thm:INDtoGHD} and $s = C t_{j^*}$ for $t_j = 1/(2^{2j} \epsilon^2)$, it follows that Bob can solve {\sf IND} by estimating 
each of $\phi_{\ell^*+1}$ and $\phi_{\ell^*+2}$ up to additive error $\frac{\gamma}{\eps 2^{j-1}}$, for $\gamma > 0$ a sufficiently small constant depending on $C, \alpha,$ and $\beta$. 

Let $\hat{\phi}$ be an estimate of $\phi$ with  $L_1$ error at most $\eps'm$ for an $\eps' = \Theta(\eps)$ to be determined. For this construction, the total number of stream updates performed by both Alice and Bob is at most:
\begin{equation}
    m = \left[\sum_{j = 1}^r O(2^{2j}) C/(2^{2j} \epsilon^2)\right] = O\left(\frac{\log (1/\eps)}{\eps^2}\right).
\end{equation}
Note for this calculation, it is crucial that Bob only inserts elements for $k = k^*$. Indeed, since Bob inserts $\Theta(2^j)$ elements for each $\ell \in [2^{j-1} +1, 2^j]$ which is a multiple of $4$ and for a given value of $k^*$, this is already $\Theta(2^{2j})$ stream elements. Had Bob inserted $\Theta(2^j)$ elements instead for each $k \in [2^{j-1}+1, 2^j]$, this would be a much larger $\Theta(2^{3j})$ total stream elements for a given $j$, which would cause $m$ to be too large for our error guarantees. 

Now we pick $\epsilon' = \Theta(\epsilon)$ such that the total error is $\eps' m = \frac{\gamma \log(1/\eps)}{100 \eps}$. 
For each $j \in [0, \ldots, \log(1/\eps) -1]$, let the $j$th ``bucket'' $B_j = [2^{j-1} + 1, \ldots, 2^{j}]$ be the corresponding geometric scale of frequencies indexed by $j$. As there are $\log (1/\eps)$ buckets, for at least $9/10$ fraction of the buckets, the $L_1$ error incurred on that bucket (i.e., between $\phi_{B_j}$ and $\hat{\phi}_{B_j}$) is at most $\frac{\gamma}{10\eps}$, as otherwise, the total error would exceed $\frac{\gamma \log (1/\eps)}{100\eps}$.

Let $B_j$ be any such bucket where the error incurred on that bucket is at most $\frac{\gamma}{10\eps}$. As $|B_j| = 2^{j-1}$, for a $9/10$ fraction of the indices $\ell \in B_j$, $|\phi_{\ell} - \hat{\phi}_{\ell}| \leq \frac{\gamma}{\eps 2^{j-1}}$. In particular, for a $4/5 = 1-2 \cdot 1/10$ fraction of the $\ell$, both $\phi_{\ell+1}$ and $\phi_{\ell+2}$ are estimated to within error $\frac{\gamma}{\eps 2^{j-1}}$, which is the error needed by Bob as described above. The crucial point now is that the state of the streaming algorithm does not depend on $j^*$ or $\ell^*$, and so since our input distribution $\mu$ for {\sc IND} is uniform, it follows that this is the error Bob has on each of $\phi_{\ell^*+1}$ and $\phi_{\ell^*+2}$ with failure probability at most $1/10 + 1/5$. 

It follows from the previous paragraph that since our input distribution $\mu$ for {\sc IND} is uniform, and also since our randomized reduction on the $(j^*, \ell^*, k^*)$-th instance fails with probability at most $1/100$, and also since our $L_1$ profile streaming algorithm succeeds with probability at least $99/100$, 
that the total success probability for Bob in solving the {\sf IND} problem over all randomness is at least:
$$1-1/10 -1/5 - 1/100 -1/100 > 2/3,$$
and so it follows by Theorem \ref{thm:IND} that Alice's message, which is the state of the streaming algorithm for the $L_1$ profile problem, must be $\Omega\left(\frac{\log (1/\eps)}{(\eps')^2}\right) = \Omega\left(\frac{\log(1/\eps)}{\eps^2}\right)$ bits long, which completes the proof of the following theorem.

\end{proof}

\subsection{Reduction for error type \texorpdfstring{$D$}{D}}

\begin{theorem}\label{thm:lb_D}
    Any single pass streaming algorithm which outputs an additive $\epsilon D$ approximation to $\phi_1$ with success probability at least $9/10$ requires $\Omega(1/\eps^2)$ bits of memory. 
\end{theorem}

\begin{proof}
We reduce the $L_1$ profile problem with the guarantee in \cref{e:l1_D} with $\tau=1$ from the {\sc IND} problem (via the {\sc Gap-Hamming} problem as in \cref{thm:INDtoGHD}).
Using \cref{thm:INDtoGHD}, have Alice and Bob construct $s$-bit strings $w$ and $z$ from their respective inputs $x \in \{0,1\}^n$ and $y \in [n]$. Recall that $s = \Theta(n)$ and choose $n=\Theta(1/\eps^2)$ such that $s = 1/\eps^2$.

For each $i \in [s]$, Alice inserts one copy of $i$ into the stream iff $w_s = 1$, and she sends the state of the streaming algorithm to Bob.
Then, Bob does the same: for each $i \in [s]$, he inserts one copy of $i$ into the stream iff $z_i = 1$.
Let $\phi$ be the true profile of the stream. The number of elements of frequency $1$, $\phi_1$, is exactly the number of elements which appear in exactly one of $w$ or $z$, so $\phi_1 = \Delta(w, z)$.

By \cref{thm:INDtoGHD}, for some constant $\gamma$, an estimate of $\phi_1$ up to $\pm \frac{\sqrt{s}}{\gamma}$ would, with probability $99/100$, solve the original {\sc IND} problem.
Note that $D \leq s = 1/\eps^2$.
Therefore, an algorithm which, with probability $9/10$, produces an estimate of $\phi_1$ up to additive error $\frac{\eps D}{\gamma} \leq \frac{1}{\eps \gamma} = \frac{\sqrt{s}}{\gamma}$ would solve the {\sc IND} problem with probability $89/100$.
\end{proof}

\section*{Acknowledgements}

Justin Chen is supported by an NSF Graduate Research Fellowship under Grant No.\ 174530.
Piotr Indyk is supported by the NSF TRIPODS program (award DMS-2022448) and the Simons Investigator Award. David Woodruff is supported in part by a Simons Investigator Award. 

\bibliographystyle{alpha}
\bibliography{bib}

\newpage
\appendix

\section{Re-analyzing the Algorithm of Datar and Muthukrishnan}
\label{s:basic}

In this section, we re-analyze the  Datar-Muthukrishnan algorithm in the case where the error function is measured using the $L_1$ norm according to  \cref{e:l1_m} where our goal is to estimate the profile up to error $\pm \eps m$. 

Recall that the algorithm selects the set $S$ of $s$ samples uniformly at random from the support of the stream and computes exact frequencies of each element sampled. Let $D$ be the number of distinct elements in the stream and let $F_i$ be the number of samples with frequency $i$. 
The algorithm then estimates the ratio $\phi_i/D$ by $\frac{1}{s} F_i$.

In our case, we make the following small modifications. First, since our goal is to estimate $\phi_i$ as opposed to $\phi_i/D$, we run a distinct elements streaming algorithm~\cite{kane2010optimal}  in parallel to get an estimate $\hat{D}$ of the number of distinct elements up to a $(1 \pm \eps)$ multiplicative factor. This requires $O(1/\eps^2 + \log n)$ space, which is subsumed by the overall space bound. Then, our estimate of $\phi_i$ will be $\hat{\phi}_i = \frac{\hat{D}}{s} F_i$. Since replacing $D$ by $\hat{D}$ changes the estimates by only a $(1 \pm \eps)$ multiplicative factor, in the rest of this section we assume the algorithm knows $D$ exactly. 

Second, for all $i > 2/\eps$, we set $\hat{\phi}_i$ to zero instead of $\frac{D}{s} F_i$. (As we note below, this adds only $O(\eps m)$ to the total error bound). We refer to this procedure as the {\em modified Datar-Muthukrishan algorithm}.

\begin{theorem}\label{thm:dm-samples}
The modified Datar-Muthukrishan algorithm, with sample size $s = O(1/\eps^2 \log(1/\eps))$, returns an estimated profile vector $\hat{\phi}$ such that $\|\phi - \hat{\phi}\|_1 \leq \eps m$ with constant probability (say, at least $2/3$).
\end{theorem}

\begin{proof}
First, we note that estimating all $\phi_i$ for $i > 2/\eps$ as zero contributes at most $\eps m/2$ to the $L_1$ error of the estimate. This is because the total number of elements with frequency at least $2/\eps$ is at most $\eps m/2$.

\cref{lem:empiricalprofile2} bounds the $L_1$ error of the reweighted empirical profile. In particular, if $s = \Omega\left(\frac{D \log(1/\eps)}{m \eps^2}\right)$,
\begin{equation*}
    \sum_{i=1}^{\lceil 2/\eps \rceil} \left| \phi_i - \frac{D}{s} F_i \right| = O(\eps m).
\end{equation*}
As $D \leq m$, $s = \Theta(1/\eps^2 \log(1/\eps))$ suffices to achieve expected error of, say, $\eps m/6$. Markov's inequality completes the proof.
\end{proof}

We note that the above algorithm, for each sample, stores $\log n + \log m$ bits to maintain the identity of the sample as well as its count for a total space complexity $O(1/\eps^2 \log(1/\eps) \log(nm))$ bits. In the following section, we present an algorithm which uses $O(1/\eps^2 \log^2(1/\eps) + \log n)$ bits, avoiding this multiplicative $\log (nm)$ dependence.

\subsection*{Improving Storage through Hashing}
\label{s:compression}
In this section we outline an improved algorithm with a reduced space bound of $O(1/\eps^2 \log^2(1/\eps) + \log n)$, i.e., replacing $\log (nm)$ with $\log(1/\eps)$. Although this bound is still not optimal, the algorithm in this section will help us illustrate the challenges in obtaining the optimal bound.

 First, recall from the analysis in the previous section that a more fine-tuned bound on the number of required samples is $s = O(1/\eps^2 \log(1/\eps) \cdot D/m)$, so as the number of distinct elements decreases, we need fewer samples.
 Let $C, C'>1$ be sufficiently large constants. 
Consider the following algorithm for processing a stream element $x_i$ given $\eps, m$ and two hash functions $g$ and $h$:
\begin{enumerate}
    \item \textbf{Sampling:}  Hash $x_i$ using $g$ to the universe $[\frac{m}{C'/\eps^2 \log(1/\eps)}]$  uniformly at random.
    \begin{itemize}
        \item If $g(x_i) = 1$, continue to step 2.
        \item Else, ignore $x_i$.
    \end{itemize}
    \item \textbf{Compression:} Hash $x_i$ using $h$ to the universe $[C (1/\eps^2 \log(1/\eps))^2]$
    \begin{itemize}
        \item Insert/increment the count of $h(x_i)$ in a dictionary (sparse hash table) $H$
        \item If the count of $h(x_i)$ exceeds $2/\eps$, marks its count as N/A and ignore it going forward
    \end{itemize}
\end{enumerate}
Call $S$ the set of unique stream elements that hash to $1$ under $g$. After all items are inserted, we use the dictionary $H$ to compute an empirical estimate $\hat{\phi}$ of the profile (rescaled by $D/|S|$), as in the previous algorithm. In what follows, we show that the resulting estimated profile $\hat{\phi}$ is within an $L_1$ distance of $\eps m$ from the true profile vector $\phi$,  with constant probability.

First, we will show that with high constant probability, $|S|$ is at least $1/\eps^2 \log(1/\eps) D/m$. The probability that any element lands in $S$ is $\frac{C'/\eps^2 \log(1/\eps)}{m}$ and there are $D$ elements, so the expected number of samples is as desired. As the distribution of $|S|$ is binomial, the variance is at most the expectation and therefore the size of the sample is correct up to constant factors with  constant probability via Chebyshev's  inequality.

Now, consider the hash table $H$. If there are no collisions in the hash table, its nonzero entries (those stored in the dictionary) will contain the true counts of each element in $S$ (ignoring elements with counts exceeding $2/\eps$, which is fine since we are not counting those). Note that $|S| < 10C'/\eps^2 \log(1/\eps)$ with a large constant probability as $D \leq n$, so $|H| = C/(10C')^2 |S|^2$. For large enough constant $C$, it is unlikely for there to be any collisions, and $H$ contains the appropriate number of samples along with their true counts.

\paragraph*{Space} It takes $\log n$ bits to store $x_i$. The rest of the space is taken up by the hash table. There are $|S|$ elements in the hash table, which we have already argued is at most $O(1/\eps^2 \log(1/\eps))$. For each element, we must store its identity $h(x_i)$ as well as the corresponding count, both of which take up $O(\log(1/\eps))$ bits. In our analysis, we only ever require pairwise independence, so the hash functions can be stored in $O(\log n)$ space.
Therefore, the total space in bits of this construction is
\begin{equation*}
    O(1/\eps^2 \log^2(1/\eps) + \log n)
\end{equation*}

\paragraph*{Improving the bound} It can be seen that the ``extra" $\log(1/\epsilon)$ factor in the $O(1/\eps^2 \log^2(1/\eps))$ bound is mostly due to the need for avoiding collisions of the sampled elements $S$ under the hash function $h$ (i.e., ensuring that $h$ is perfect) \footnote{In addition, we also need to maintain the count of each sampled element, which also takes $O(\log(1/\eps))$ bits per sample. However, this issue can be solved more easily.}. This requires storing $O(\log(1/\eps))$ bits per sample, to disambiguate distinct elements in $S$. The algorithm presented in \cref{s:algorithm} achieves the optimal space by hashing elements in $S$ to a hash table of size $O(|S|)$, not $O(|S|^2)$, removing the need to store hashed IDs. This, however, comes at the price of allowing collisions, which means that elements with different frequencies are mixed together. Iteratively ``inverting'' this mixing process to obtain frequency estimates is the most technically challenging part of our algorithm.

\end{document}